\documentclass[
  aps,%
  reprint,%
]{revtex4-2}

\usepackage{amsmath,amssymb,amsthm}
\usepackage{bm,bbm}
\usepackage{physics}

\usepackage[sc,osf]{mathpazo}
\usepackage[T1]{fontenc}

\usepackage{xcolor}
\usepackage{graphicx}

\usepackage[bookmarks=false]{hyperref}
\hypersetup{colorlinks=true,citecolor=blue,linkcolor=red,%
urlcolor=blue,pdfstartview=FitH,bookmarksopen=true}
\theoremstyle{definition}

\newtheorem{lem}{Lemma}

\newtheorem{observation}{Observation}

\newcommand{\vect}[1]{\mathbold{#1}}

\begin{document}

\title{Scalable Certification of Entanglement in Quantum Networks}

\author{Jing-Tao Qiu}
\affiliation{Department of Physics, Shandong University, Jinan 250100, China}

\author{D. M. Tong}
\email{tdm@sdu.edu.cn}
\affiliation{Department of Physics, Shandong University, Jinan 250100, China}
\date{\today}

\author{Xiao-Dong Yu}
\email{yuxiaodong@sdu.edu.cn}
\affiliation{Department of Physics, Shandong University, Jinan 250100, China}

\begin{abstract}
Quantum networks form the backbone of long-distance quantum information
processing. Genuine multipartite entanglement (GME) serves as a key indicator of
network performance and overall state quality. However, the widely used methods
for certifying GME suffer from a major drawback that they either detect only a
limited range of states or are applicable only to systems with a small number of
parties. To overcome these limitations, we propose a family of sub-symmetric
witnesses (SSWs), which are tractable both theoretically and experimentally.
Analytically, we establish a connection between SSWs and the cut space of graph
theory, enabling several powerful detection criteria tailored to practical
quantum networks. Numerically, we show that the optimal detection can be
formulated as a linear program, offering a significant efficiency advantage over
the semidefinite programs commonly employed in quantum certification.
Experimentally, SSWs can be evaluated via local measurements, with resource
requirements independent of the local dimension in general, and even
independent of the overall network size in many practical networks.
\end{abstract}

\maketitle

%\section{Introduction}
\textit{Introduction.---}%
For future long-range quantum information processing tasks, quantum networks
featuring the entanglement shared among different parties are the promising
candidate, and have been successfully established in laboratory and urban
settings \cite{PeevSECOQCQuantumKey2009, SasakiFieldtestquantum2011,
MaoIntegratingquantumkey2018}.
In practice, quantum networks allow remote quantum devices to operate as a unified system, supporting applications such as secure quantum communication \cite{Scaranisecuritypracticalquantum2009, XuSecurequantumkey2020, PirandolaAdvancesquantumcryptography2020},
distributed quantum computation \cite{CacciapuotiQuantumInternetNetworking2020, BarralReviewDistributedQuantum2025}, clock synchronization \cite{Komarquantumnetworkclocks2014, Nicholelementaryquantumnetwork2022}, and quantum sensing \cite{GeDistributedQuantumMetrology2018, GuoDistributedquantumsensing2020, ZhangDistributedquantumsensing2021}. 
In addition to their technological importance, quantum networks provide a
fertile platform for advancing quantum foundations. Networked configurations
introduce independence constraints and causal structures that do not appear in
standard Bell scenarios, giving rise to richer forms of correlations and new
manifestations of nonlocality \cite{BranciardCharacterizingNonlocalCorrelations2010, BranciardBilocalversusnonbilocal2012, FritzBellstheoremcorrelation2012, TavakoliBellnonlocalitynetworks2022, PozasKerstjensFullNetworkNonlocality2022}.

A central challenge in advancing quantum networks is the reliable certification
of quantum resources across complex topologies
\cite{HorodeckiQuantumentanglement2009, GuehneEntanglementdetection2009,
FriisEntanglementcertificationtheory2019,
NavascuesGenuineNetworkMultipartite2020, WolfeQuantumInflationGeneral2021}.
As networks grow in size, ensuring that the distributed state retains useful
quantum features, such as entanglement or nonlocality, becomes increasingly
difficult \cite{Kimblequantuminternet2008, WehnerQuantuminternetvision2018}.
Among these resources, genuine multipartite
entanglement (GME) plays a fundamental role. Unlike bipartite correlations, GME
captures irreducible quantum correlations shared across the entire network and
is essential for tasks such as measurement-based computation \cite{RaussendorfOneWayQuantum2001, RaussendorfMeasurementbasedquantum2003, BriegelMeasurementbasedquantum2009}, network
teleportation \cite{PirandolaAdvancesquantumteleportation2015,
HermansQubitteleportationnon2022}, and multi-party cryptographic protocols
\cite{HilleryQuantumsecretsharing1999, CleveHowShareQuantum1999,
ChenMulti‑partitequantumcryptographic2007, MurtaQuantumConferenceKey2020}. A
deeper understanding of GME therefore
directly supports the efficient implementation of these network-level quantum
information processing tasks.

However, the characterization of quantum entanglement is intrinsically
challenging \cite{HorodeckiQuantumentanglement2009,
GuehneEntanglementdetection2009, FriisEntanglementcertificationtheory2019}. Even
in the bipartite setting, deciding whether a quantum state is entangled is
NP-hard \cite{GurvitsClassicaldeterministiccomplexity2003}. The complexity grows
substantially in the multipartite regime, where entanglement structures become
significantly richer and more difficult to certify.
For GME in particular, two widely used approaches are fidelity-based
entanglement witnesses \cite{BourennaneExperimentalDetectionMultipartite2004}
and semidefinite-program (SDP) relaxations based on the
positive partial transpose (PPT) criterion \cite{JungnitschTamingMultiparticleEntanglement2011}. However, neither method extends
naturally to large quantum networks.
Fidelity-based witnesses are simple to construct and experimentally friendly,
but they can certify only a narrow class of target states \cite{WeilenmannEntanglementDetectionMeasuring2020}. SDP-based PPT
relaxations, while capable of detecting a broader range of entangled states,
suffer from rapid growth in computational cost with increasing network size and
local dimension. As a result, these SDPs become intractable beyond small systems
containing only a few qubits.

In this work, we introduce a family of sub-symmetric witnesses (SSWs) that
provide a scalable approach to certifying and even quantifying GME
in quantum networks. Theoretically, we establish a connection between SSWs and
the cut space of graph theory, enabling detection criteria that naturally align
with realistic network topologies. Computationally, we show that optimizing
SSW-based detection reduces to a linear program, offering a substantial
efficiency advantage over the semidefinite programs commonly used in
entanglement certification.
Experimentally, SSWs can be evaluated via local measurements, with resource
requirements independent of the local dimension in general, and even
independent of the overall network size in many practical networks.
To illustrate the strength of our approach, we
demonstrate that for quantum networks with randomly generated topologies, SSWs
significantly outperform existing GME certification methods.

%\section{Preliminaries}
\textit{Preliminaries.---}%
The celebrated protocol of quantum teleportation demonstrates that quantum
information can be transmitted using shared entanglement together with classical
communication \cite{BennettTeleportingunknownquantum1993, ifmmodeZelseZfiukowskiEventreadydetectors1993}. This insight motivates the development of quantum
networks, in which multiple parties are interconnected by shared quantum states,
enabling distributed quantum information processing
\cite{Kimblequantuminternet2008, WehnerQuantuminternetvision2018}.
The most commonly shared states are Bell states, but the distribution process is
inevitably affected by noise \cite{BennettPurificationNoisyEntanglement1996}. Under the independence assumption on the
noise, the network states can be written as the tensor product of these 
bipartite states, together with the network topology represented by an
undirected graph $G=(V,E)$, where the vertices $V$
represent the parties and the edges $E$ represent the shared states as
demonstrated in Fig.~\ref{fig:networks}. Without loss of generality, all
networks considered in this work are assumed to be connected.

For example, the bilocal network in Fig.~\ref{fig:networks}(a) has three
parties $A$, $B$, and $C$ \cite{BranciardCharacterizingNonlocalCorrelations2010}. Parties $A$ and $B$ share a bipartite state
$\rho_{AB_1}$, and parties $B$ and $C$ share a bipartite state $\rho_{B_2C}$.
Therefore, the overall state of the network is
$\rho=\rho_{AB_1}\otimes\rho_{B_2C}$, which gives rise to multipartite
correlations among the three parties $A$, $B=B_1B_2$, and $C$.
Note that the Hilbert space dimensions of different parties in the network may
differ. In this particular case, we have $\dim(B)=\dim(A)\dim(C)$. To simplify
notation, we hereafter omit the subscripts indicating local subsystems of
the network states.
For instance, the network state in Fig.~\ref{fig:networks}(a) is written as
$\rho=\rho_{AB}\otimes\rho_{BC}$, and that in Fig.~\ref{fig:networks}(b)
as $\rho=\rho_{AB}\otimes\rho_{BC}\otimes\rho_{BD}$.

The principal goal of this work is to develop an efficient method for
determining the overall entanglement properties of network states.
Formally, a state of $n$ parties is said to exhibit GME if it cannot be
written as a mixture of states separable with
respect to some bipartition of the parties, otherwise it is called biseparable. 
For example, the state $\rho=\rho_{AB}\otimes\rho_{BC}$ in the bilocal network
is GME if and only if the state is not a mixture of states $\rho_{A|BC}$,
$\rho_{B|AC}$, $\rho_{C|AB}$, which are separable with respect to the partitions
$A|BC$, $B|AC$, $C|AB$, respectively.

\begin{figure}
    \includegraphics[width=0.45\textwidth]{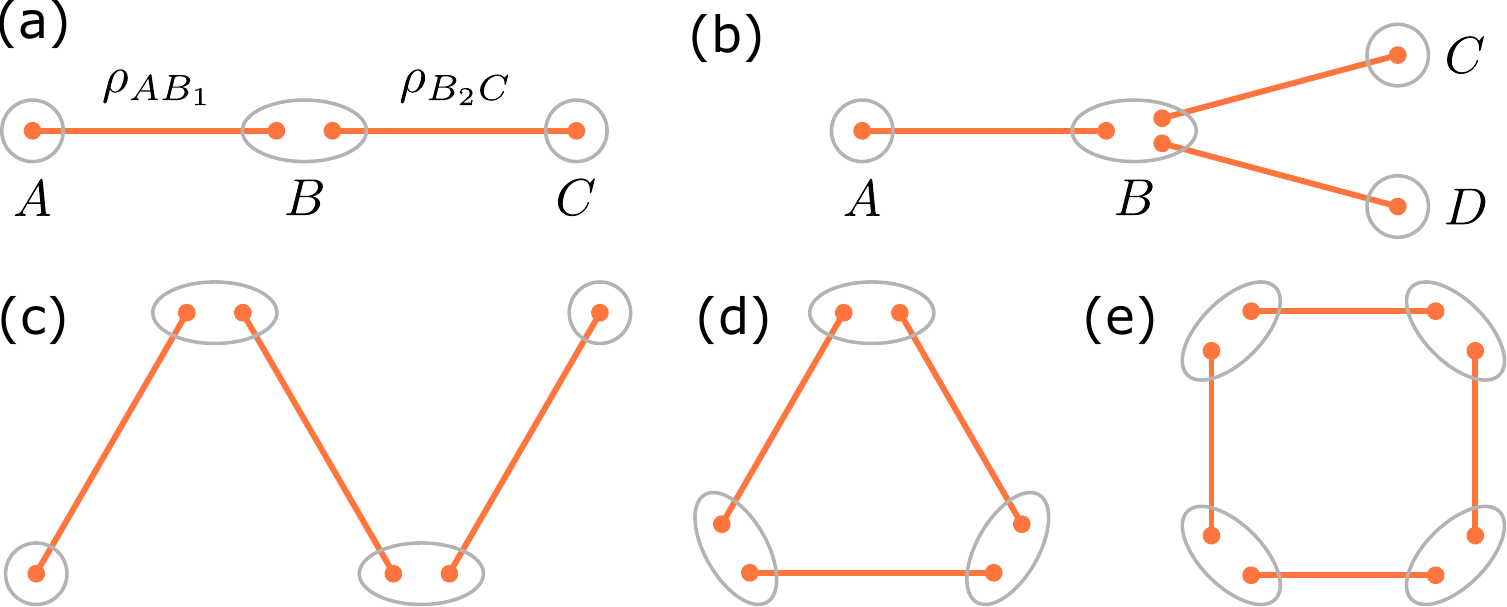}
    \caption{
    Networks with different topologies. (a)-(c) are tree networks, (d) and (e) are ring networks.
    }
	\label{fig:networks}
\end{figure}

For the certification of GME, the most widely used approach is
entanglement witnesses
\cite{HorodeckiQuantumentanglement2009, GuehneEntanglementdetection2009,
FriisEntanglementcertificationtheory2019}.
Mathematically, a GME witness is an observable
$W$ such that $\Tr(W \rho_s)\ge0$ for all biseparable states $\rho_s$, thus
certifying the existence of GME in $\rho$ when $\Tr(W \rho)<0$.
In practice, this is often implemented via fidelity estimation,
where $\rho$ is certified as entangled if its fidelity with a target
pure state $\ket{\psi}$ exceeds the maximum fidelity achievable by any
biseparable state $\alpha:=\max_{\rho_s}\ev{\rho_s}{\psi}$. This corresponds to
the GME witness $W=\alpha\mathbb{I}-\ketbra{\psi}$.
The advantage of fidelity-based GME witnesses is that the threshold value
$\alpha$ can be derived analytically,
while the drawback is that the set of
certifiable entangled states is severely
restricted \cite{WeilenmannEntanglementDetectionMeasuring2020}.

An alternative method is based on the PPT
relaxation of separable states, where it suffices to show that $W$ is a fully
decomposable witness \cite{JungnitschTamingMultiparticleEntanglement2011}, i.e.,
for every bipartition of the global system $S|S^c$
there exist some positive semidefinite $P_S, Q_S$ such that
$W=P_S+Q_S^{\Gamma_S}$, where $S^c=V\setminus S$ denotes the complement of $S$, and $\Gamma_S$ denotes the
partial transpose of subsystems $S$.
This observation allows the certification of GME
to be recast as an SDP \cite{JungnitschTamingMultiparticleEntanglement2011}.
Moreover, fully decomposable witnesses also provide a quantitative
characterization of the GME. We define the GME negativity $\mathcal{N}(\rho)$
as
\begin{equation}\label{eq:GME_negativity}
	\begin{aligned}
    \min_{p_S,\rho_S}: \quad & \sum_Sp_S\mathcal{N}_S(\rho_S)\\
    \text{s. t.}: \quad & \rho=\sum_Sp_S\rho_S,\\
    & p_S\ge 0,~\rho_S\succeq 0,~\Tr(\rho_S)=1,~\forall S,
	\end{aligned}
\end{equation}
where $\mathcal{N}_S(\rho_S)$ is the bipartite negativity with respect to the
bipartition $S|S^c$,
defined as the absolute sum of the negative eigenvalues of
$\rho_S^{\Gamma_S}$ \cite{VidalComputablemeasureentanglement2002,
PlenioLogarithmicNegativityFull2005},
and $\forall S$ refers to all nontrivial subsets $S$ of $V$.
In Appendix~\ref{app:GME_negativity}, we prove that $\mathcal{N}(\rho)$ is a
legitimate measure of
GME, and it can also be obtained by solving the following SDP:
\begin{equation}\label{eq:GME_negativity_dual}
	\begin{aligned}
		\max_{W,Q_S}: \quad& -\Tr(W\rho)\\
		\text{s. t.}: \quad
		& 0 \preceq Q_S \preceq \mathbb{I},~
        Q_S^{\Gamma_S}\preceq W,~\forall S.\\
	\end{aligned}
\end{equation}
The PPT-based approach offers a substantial advantage over the fidelity-based
method by detecting a much larger set of GME states. However, it
suffers from poor scalability: the scale of the associated SDP increases rapidly
with the network size and the local dimension of subsystems. Consequently, its
utility for quantum networks is limited, as it is only applicable to
quantum systems containing several qubits in practice.

In this work, we introduce the sub-symmetric witness (SSW) for the certification
of GME in quantum networks, which integrates the advantages of both
fidelity-based and PPT-based approaches. Our approach is applicable to
large-scale quantum networks and capable of certifying a broad class of
entangled network states.

%\section{Sub-symmetric witnesses}
\textit{Sub-symmetric witnesses.---}%
As previously noted, the network state can be modeled by an undirected graph
$G=(V,E)$, where each vertex $v\in V$ represents a party and each edge $e\in E$
corresponds to a shared state $\rho_e$. Rather than restricting $\rho_e$ to
qubit systems, we allow them to be general qudit states. This generalization
ensures that our results remain applicable to practical scenarios involving the
distribution of multi-copy or high-dimensional entangled states.
For notational convenience, we assume that all shared states $\rho_e$ 
have an identical local dimension, though our results generalize
straightforwardly to heterogeneous systems.

The central idea of our approach is to estimate the fidelities of all shared
states $\rho_e$ while simultaneously accounting for the correlations among these
estimations. This procedure yields a GME witness that is tractable both
theoretically and experimentally.
One remarkable result for the fidelity‐based witness is that it suffices to
carry out fidelity estimations with respect to maximally entangled states \cite{GuehneGeometryFaithfulEntanglement2021}.
Therefore, we define an SSW as a fully decomposable observable of the form
\begin{equation}\label{eq:SSW}
	W=\sum_{T\subseteq E}w_T\textstyle\bigotimes_{e\in T}\Phi_e\otimes \bigotimes_{e\notin T}(\mathbb{I}_e-\Phi_e),
\end{equation}
where $w_T$ are real numbers, 
and $\Phi_e=\ketbra{\phi^+}$ with $\ket{\phi^+}
=\frac{1}{\sqrt{d}}\sum_{i=0}^{d-1}\ket{ii}$   
denoting the maximally entangled state associated with edge $e$.
The term sub-symmetric denotes the property that the partial transposition on
either subsystem of an edge yields an identical result.
Henceforth, we treat these operations as equivalent and refer simply to the
partial transposition of an edge.

The sub-symmetric property enables us to characterize the partial transposition 
of parties $S$ by the corresponding set of partially 
transposed edges, denoted $T(S) \subseteq E$. For instance, in 
Fig.~\ref{fig:networks}(a), when party $\{A\}$ is partially transposed, 
the edge $(A,B)$ undergoes partial transposition, i.e., for $S=\{A\}$ we obtain 
$T(S)=\{(A,B)\}$. When parties $\{A,B\}$ are partially transposed, the edge 
$(A,B)$ is partially transposed twice and is therefore invariant, while the edge
$(B,C)$ is partially transposed once, resulting in $T(S)=\{(B,C)\}$. Generally,
we have
\begin{equation}
    T(S)=\qty{(u,v) \in E \mid u\in S, v\notin S \text{ or } u\notin S, v\in S}
    \label{eq:defTS}
\end{equation}
for the nontrivial subset $S$ of $V$.
These observations considerably simplify the characterization of SSWs. The
following theorem provides the foundation for all subsequent theoretical and
numerical analysis; see Appendix~\ref{app:lp} for the proof.
\begin{observation}\label{obs:lp}
	The observable $W$ in the form of Eq.~\eqref{eq:SSW} is an SSW if and only if
    for all subsystems $S$ there exist
    $q_T(S)\le w_T$ such that
    \begin{equation}\label{eq:linearConstraints}
        \sum_{T\subseteq T(S)} q_{T\cup\tilde{T}}(S) \mathbold{x}_T(S) \ge 0, \quad \forall \tilde{T}\subseteq E\setminus T(S),
    \end{equation}
    where 
    \begin{equation}
        \mathbold{x}_T(S) = \bigotimes_{e\in T} (1, -1) \otimes
        \bigotimes_{e\in T(S)\setminus T} (d-1, d+1),
    \end{equation}
    and the inequality is understood componentwise.
\end{observation}

A direct application of Observation~\ref{obs:lp} is that  the certification of
GME can be formulated as a linear program (LP) instead of an SDP. Remarkably,
the number of variables in the LP is independent of the local dimension $d$, so
the complexity does not increase in practical scenarios involving the
distribution of multi-copy or high-dimensional entangled states. Moreover, we
show that the negativity of GME can be lower bounded by the following LP,
with the bound being tight when the state distribution process is
subject to white noise:
\begin{equation}\label{eq:negativity}
	\begin{aligned}
		\max_{w_T,q_T}: \quad& -\sum_{T \subseteq E}w_Tp_T\\
		\text{s. t.}: \quad & 
		q_T(S) \le w_T,\\
		&0\le\sum_{T\subseteq T(S)} q_{T\cup\tilde{T}}(S) \vect{x}_T(S)\le d^{\abs{T(S)}},\\
	\end{aligned}
\end{equation}
where the constraints run over all $T \subseteq E$, $\tilde{T}\subseteq
E\setminus T(S)$, and $S \ne \emptyset \text{ or }V$,
$\abs{T(S)}$ denotes the number of edges in $T(S)$,
and $p_T=\prod_{e\in T}\Tr(\rho_e\Phi_e)
\prod_{e\notin T}(1-\Tr(\rho_e\Phi_e))$
\footnote{This also applies to general $\rho$ instead of $\rho=\bigotimes_{e \in
E} \rho_e$, where $p_T=\Tr(\rho \qty[\bigotimes_{e \in T} \Phi_e \otimes
\bigotimes_{e \notin T}(\mathbb{I}_e-\Phi_e)])$.}.

Modern LP solvers are capable of handling instances with over a billion variables \cite{ApplegatePDLPPracticalFirst2025},
which enables the application of our approach to networks comprising several
dozen parties. The result can be further strengthened by exploiting the
topological structure of the quantum network.
On the one hand, for networks with specific topologies, such as tree
and ring networks, the numbers of variables and constraints scale only polynomially
with the size of the network, which makes the numerical algorithms much more
efficient. On the other hand, we establish a set of general SSWs determined by
the network topology, which eliminates the need for optimization.
These simplifications are grounded in a graph-theoretical analysis of SSWs.

%\section{Graph-theoretical analysis of SSWs}
\textit{Graph-theoretical analysis of SSWs.---}%
Observation~\ref{obs:lp} provides a complete characterization of SSWs, from
which it is evident that the constraints in Eq.~\eqref{eq:linearConstraints} are
entirely determined by $T(S)$. In graph theory, $T(S)$ is the so-called cut
for $S$. Together with $\emptyset$, all cuts $K$ form a set
$\mathcal{B}(G)$ called the cut space of $G$.
Crucially, the cut space constitutes a vector space over the two-element field
$\mathbb{F}_2$ and serves as the orthogonal complement of the cycle space
\cite{DiestelGraphtheory2025}.

The main motivation for considering the cut space $\mathcal{B}(G)$ rather than
all subsystems $S$ is that the cut space reveals additional symmetries beyond
those of the underlying graph $G$.
In addition, the redundancies arising from the equivalence of partial transposes
with respect to $S$ and $S^c$ are automatically removed in
$\mathcal{B}(G)$.
With the cut space, both Observation~\ref{obs:lp} and
Eq.~\eqref{eq:negativity} can be simplified accordingly.
In particular, Observation~\ref{obs:lp} can be reformulated as follows:
for all nonempty $K\in\mathcal{B}(G)$, there exist $q_T(K)\le w_T$ such that
\begin{equation}\label{eq:linearConstraintsK}
  \sum_{T\subseteq K} q_{T\cup\tilde{T}}(K)\vect{x}_T(K)\ge 0,
  \quad \forall\tilde{T}\subseteq E\setminus K,
\end{equation}
with  $\vect{x}_T(K)=\bigotimes_{e\in T} \qty(1,-1) \otimes \bigotimes_{e \in K\setminus T}\qty(d-1,d+1).$

Recall that in our model each edge $e \in E$ corresponds to a bipartite quantum
state $\rho_e$. Based on these $\rho_e$, networks with different topologies can
be constructed
\footnote{Due to the sub-symmetric property of SSWs, we can always assume that
$V\rho_eV^\dagger=\rho_e$, where $V$ is the swap operator.}.
For example, with three edges there are five distinct
topologies, as illustrated in Fig.~\ref{fig:differentTopology}. 
We denote by $\mathcal{G}_N$ the set of all graphs with edge set
$\{1,2,\ldots,N\}$.
The cut space naturally induces a partial order on $\mathcal{G}_N$.
Concretely, we say that $G_1 \prec G_2$ whenever the cut space of
$G_1$ is contained in that of $G_2$, i.e., $\mathcal{B}(G_1) \subseteq
\mathcal{B}(G_2)$. This, in turn, gives rise to an equivalence relation:
two graphs $G_1,G_2\in\mathcal{G}_N$ are considered equivalent, written $G_1
\sim G_2$ whenever $\mathcal{B}(G_1) = \mathcal{B}(G_2)$.

With these definitions in place, we proceed to establish several key
observations that formally characterize the relationships among SSWs, network
topology, and network entanglement.
The detailed proofs are provided in Appendix~\ref{app:property}.

\begin{observation}\label{obs:order}
Let $G_1, G_2 \in \mathcal{G}_N$.\\
(a) If $G_1 \prec G_2$, any SSW for $G_2$ is also an SSW for $G_1$;\\
(b) If $G_1 \sim G_2$, the sets of SSWs for $G_1$ and $G_2$ coincide.
\end{observation}

Recall that a graph is called a tree if it is connected and contains no cycles.
Moreover, the cut space of a graph is the orthogonal complement of its cycle
space. Hence, for a tree graph $G_\mathrm{tree}$, $\mathcal{B}(G_\mathrm{tree})$
contains all subsets of $\mathcal{G}_N$, resulting in the following observation.

\begin{observation}\label{obs:tree}
(a) Any SSW for a tree network in $\mathcal{G}_N$ is also an SSW for any other
network $G\in\mathcal{G}_N$.\\
(b) The set of SSWs for all tree networks in $\mathcal{G}_N$ coincide.
\end{observation}

\begin{figure}
    \centering
    \includegraphics[width=0.45\textwidth]{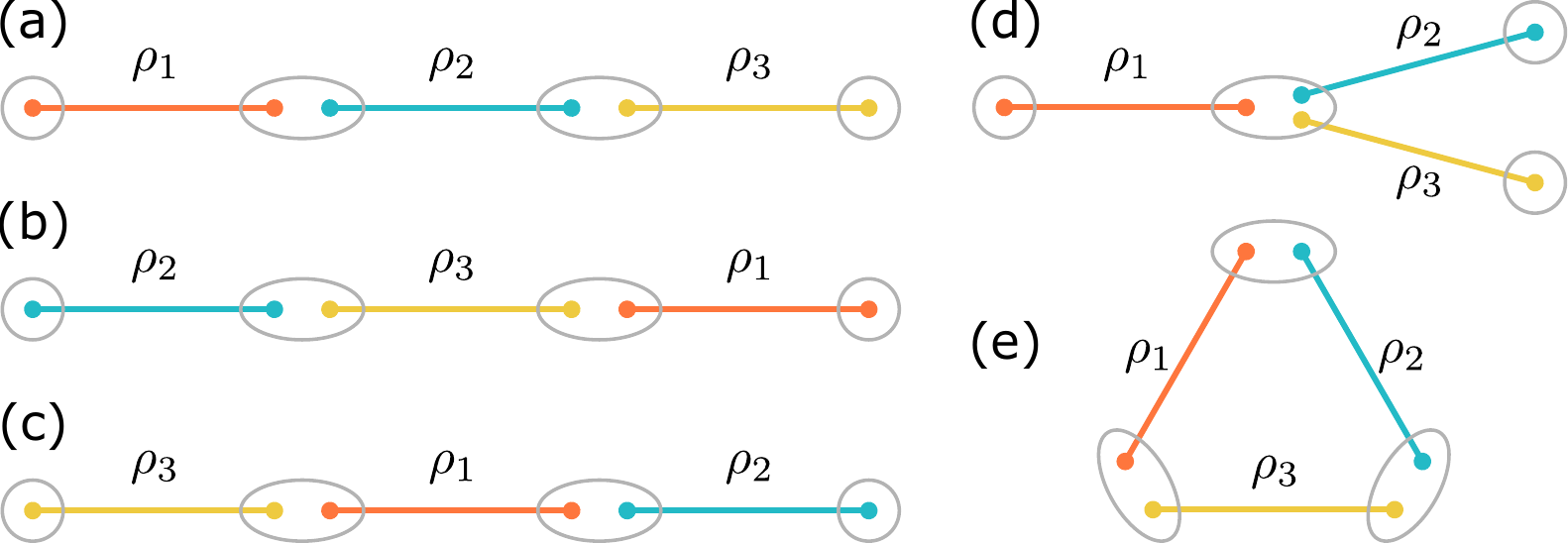}
    \caption{Five distinct topologies for a graph with three different edges
    $\rho_1,\rho_2,\rho_3$.}
    \label{fig:differentTopology}
\end{figure}

In practice, it is often a good approximation to assume that 
each $\rho_e$ is a white-noised state, i.e.,
$\rho_e=p_e\Phi+(1-p_e)\mathbb{I}/d^2$.
In this situation, Observation~\ref{obs:order} can also be formulated in a quantitative
manner: (a) If $G_1 \prec G_2$, the GME negativity of $G_1$
is not smaller than that of $G_2$;
(b) If $G_1 \sim G_2$, the GME negativity of $G_1$
is equal to that of $G_2$.
Note that constructing a connected network with $N+1$ parties
requires at least $N$ entangled pairs, which occurs precisely when
the network has a tree topology. This further leads to the following
surprising result on the entanglement of quantum networks.

\begin{observation}\label{obs:tree_negativity}
    Suppose that an $(N+1)$-party quantum network is constructed by distributing
    $N$ entangled pairs, where each distributed state is a white-noised
    maximally entangled state. Then, the GME of the network quantified by the
    negativity is independent of the topology.
\end{observation}

These graph-theoretical analysis
can be applied to simplify the construction of SSWs, from both an
analytical and a numerical perspective.
The basic idea is that, given a graph $G$, one can consider
an equivalent graph $\tilde{G}$ that is significantly more symmetric. In graph
theory, this corresponds to the notion of $2$-isomorphism \cite{Whitney2IsomorphicGraphs1933}.
For instance, any tree network can be assumed, without loss of generality, to be
a star network.

Analytically, we construct two families of GME witnesses that are strong in
practical quantum networks.
The first family, originating from tree networks, is topology-independent and
referred to as the $E$-SSW:
\begin{equation}\label{eq:E-SSW}
     W_E=-d\bigotimes_{e \in E} \Phi_e + 
     \frac{d+1}{2}\bigotimes_{e\in E}\frac{\mathbb{I}_e+d\Phi_e}{d+1}.
\end{equation}
The second family, designed to fully capture the network's topology,
is referred to as the $\mathcal{B}(G)$-SSW:
\begin{equation}\label{eq:BG-SSW}
	W_{\mathcal{B}(G)} = -\bigotimes_{e\in E}\Phi_e
    + \sum_{\emptyset \ne K \in \mathcal{B}(G)}
    \bigotimes_{e\in K}\frac{\mathbb{I}_e-\Phi_e}{d-1}
    \otimes \bigotimes_{e\notin K}\Phi_e.
\end{equation}
Please see the proof in Appendix~\ref{app:SSW_family}. 
The strength of these witnesses will be demonstrated with examples in the
following section.

Numerically, the additional symmetries revealed by analyzing the cut space can
be exploited to reduce the number of variables and constraints in the
corresponding LP formulation. For example, for tree and ring networks, the
search for SSWs can be cast as an LP involving $O(\abs{E}^3)$ variables and
constraints. More general formulations and illustrative examples are provided in
Appendix~\ref{app:cactus}.

%\section{Examples}
\textit{Examples.---}%
We demonstrate the utility of SSWs by applying them to networks where all
edges are subjected to white noise of equal strength. Specifically, the edge
state is modeled by $\rho_e=p\Phi_e+(1-p){\mathbb{I}_e}/{d^2}$ with the
visibility $p$ independent of the edge $e$.
This is referred to as isotropic pair-entangled network (IPEN) states in
Refs.~\cite{ContrerasTejadaAsymptoticSurvivalGenuine2022,
LledoAsymptoticrobustnessentanglement2024}.

We begin by considering the simplest case of tree networks.
Since the cut space in tree networks contains all subsets of $E$, the
optimization problem defined in Eq.~\eqref{eq:negativity} exhibits high
symmetry. For the IPEN state, Eq.~\eqref{eq:negativity} not only
gives the exact value of GME negativity, moreover it can also be analytically
solved and is precisely expressed by $W_E$ in Eq.~\eqref{eq:E-SSW}.

\begin{observation}\label{obs:IPEN_tree}
	The IPEN state in tree networks has the GME negativity
    \begin{equation}
        \mathcal{N}(\rho)=\max\qty{0,-\Tr(W_E\rho)}.
    \end{equation}
    Especially, $\rho$ is GME whenever
    \begin{equation}\label{eq:opt_tree}
        \qty(1-\frac{\nu_d}{d})^\abs{E}>\frac{d+1}{2d},
    \end{equation}
    where $\nu_d=\frac{(d-1)(1-p)}{d-(d-1)(1-p)}$.
\end{observation}
The left-hand side of Eq.~\eqref{eq:opt_tree} converges to
zero as $\abs{E}\to\infty$ unless $p=1$. This phenomenon is termed the
asymptotic non-survival of GME, i.e., the IPEN state in tree networks is never
GME when $\abs{E}\to\infty$ unless $p=1$
\footnote{Note that Ref. \cite{ContrerasTejadaAsymptoticSurvivalGenuine2022}
establishes a stronger non-survival result for GME, formulated in terms of
biseparability rather than PPT mixtures.}.
This outcome is expected because, in a tree
network, the addition of a new edge necessarily introduces an extra party,
making it unlikely that GME can be sustained as the number of parties grows
indefinitely. However, Observation~\ref{obs:IPEN_tree} offers a less pessimistic
view by implying that the visibility of $\rho_e$ does not
need to exponentially approach one. More precisely, to ensure that the tree
network is GME, it suffices to take $1-p$ to be $O(d\abs{E}^{-1})$.
This rate, which only scales inverse-polynomially with the number of edges or parties,
represents a less demanding condition for maintaining GME and can be
further mitigated by utilizing high-dimensional entanglement.
Moreover, Observation~\ref{obs:tree} implies that an order of $O(d\abs{E}^{-1})$
suffices to ensure the presence of GME in any network topology.

\begin{figure}
    \centering
    \includegraphics[width=0.45\textwidth]{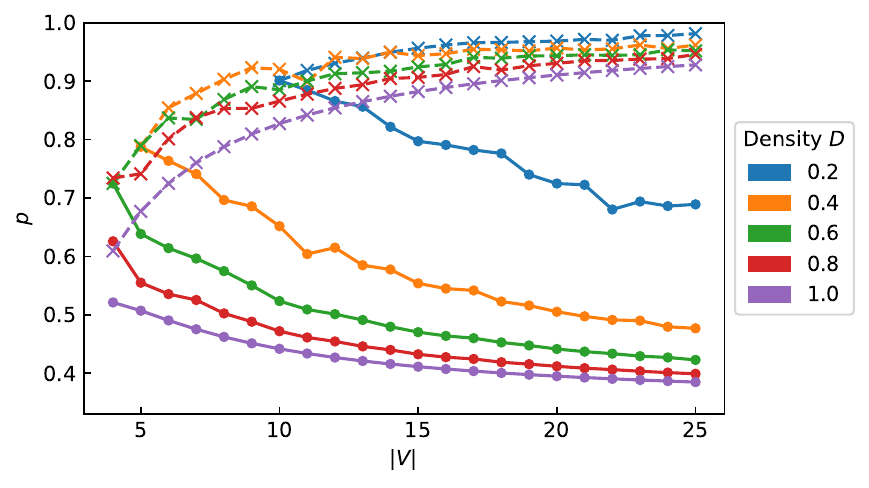}
    \caption{Average visibilities for random graphs obtained using SSWs and
    fidelity-based entanglement witnesses.
    The solid lines and dashed lines denote the SSWs and the fidelity-based witnesses, respectively.
    The visibility $p$ represents the
    lower bound above which the network state is certified as GME by the
    corresponding witness. The graph density $D$ denotes the ratio of the number
    of edges $\abs{E}$ to the maximum possible number of edges, i.e.,
    $D=\abs{E}/\binom{\abs{V}}{2}$.
    }
    \label{fig:randomgraph}
\end{figure}

In contrast to tree networks, where new edges introduce new parties, 
a strategy to enhance the GME in a quantum network is to add more edges without
increasing the number of parties.
By employing the $\mathcal{B}(G)$-SSWs defined in Eq.~\eqref{eq:BG-SSW}, we are
able to show that this strategy can guarantee the asymptotic survival is
guaranteed whenever the underlying $\rho_e$ is entangled (i.e., $p>1/(d+1)$).
Please see the proof in Appendix~\ref{app:IPEN}.
\begin{observation}\label{obs:asymptotic}
The IPEN states constructed from entangled $\rho_e$ exhibit GME as
$\abs{V}\to\infty$, provided that each party is asymptotically connected
to more than half of the remaining parties, i.e.,
$\lim_{\abs{V}\to\infty}\frac{\deg(v)}{\abs{V}}>\frac{1}{2}$ for
every $v\in V$.
\end{observation}

More generally, to demonstrate the advantage of SSWs, we compare them with the
widely used fidelity-based GME witnesses for randomly generated graphs.
Figure~\ref{fig:randomgraph} presents a comparison of the certification power
of the SSWs and the fidelity-based entanglement witnesses.
For a fair comparison, we use only the analytical bounds derived from
$\mathcal{B}(G)$-SSWs, without employing the optimal SSWs obtained
through numerical optimization.

In passing, we would like to mention that due to the form of SSW in
Eq.~\eqref{eq:SSW}, the witness value can be estimated by only determining the
fidelity with respect to the maximally entangled state. This fidelity can be
estimated efficiently using randomized local
measurements \cite{FlammiaDirectFidelityEstimation2011, ZhuOptimalverificationfidelity2019, YuOptimalverificationgeneral2019, YuStatisticalMethodsQuantum2022}, for which required copies of the states do not increase with the local dimension.
Details regarding complexity are provided in Appendix~\ref{app:complexity}.

%\section{Conclusion}
\textit{Conclusion.---}%
We proposed the SSWs for scalable entanglement certification in quantum
networks. The scalability of SSWs are shown in two aspects. On the one hand, the
problem of searching for SSWs can be cast as an LP, thus this task can be solved
systematically and more efficiently. On the other hand, SSWs can be evaluated by
estimating merely the fidelity with respect to maximally entangled state for
each shared state, making the measurement overhead scalable for large networks.
The core innovation presented is a graph-theoretical characterization of
SSWs by leveraging the cut space of the quantum network's topology.
Beyond making network entanglement certification practical,
this graph-theoretical analysis of SSWs also reveals some fundamental properties
of network entanglement, including the partial order of network topologies, the
topology-independence of negativity for tree networks, and the aggregate of
entanglement in highly connected networks.

Several interesting directions emerge from this work.
First, while the current work is devoted to the study of GME,
the scope of applicability of our method goes well beyond this context.
We expect the approach to be straightforwardly extended to certify other
relevant entanglement quantities in quantum networks, such as entanglement
depth \cite{SoerensenEntanglementExtremeSpin2001} and entanglement detection
length \cite{ShiEntanglementDetectionLength2025}, and more general entanglement
structures \cite{HuberStructureMultidimensionalEntanglement2013,
ShahandehStructuralQuantificationEntanglement2014,
CobucciDetectingdimensionalitygenuine2024}.
Second, the applicability of our method is also not limited to quantum networks.
The approach can be easily adapted to certify general multipartite entanglement
through the use of local projections \cite{HorodeckiMixedStateEntanglement1998, Clarissedistillabilityproblemrevisited2006}.
This can serve as an effective tool for the study of
the superactivation of multipartite entanglement \cite{HuberPurificationgenuinemultipartite2011,
YamasakiActivationgenuinemultipartite2022,
PalazuelosGenuinemultipartiteentanglement2022, BaksovaMulticopyactivation2025,
WeinbrennerSuperactivationIncompressibilityGenuine2025}.
Third, our current method exclusively utilizes the sub-symmetry of the maximally
entangled state. It would be worthwhile to generalize the method by
incorporating other states that exhibit sub-symmetry, such as Werner states
\cite{WernerQuantumstatesEinstein1989}.

%\begin{acknowledgments}
    This work was supported by
    the National Natural Science Foundation of China
    (Grants No.~12574537, No.~12205170, and No.~12174224)
    and the Shandong Provincial Natural Science Foundation of China
    (Grant No. ZR2022QA084).
%\end{acknowledgments}

\bibliography{NetworkWitness.bbl}
\appendix

\section{GME negativity}\label{app:GME_negativity}
\textit{GME negativity.}---%
$\mathcal{N}(\rho)$ defined in Eq.~\eqref{eq:GME_negativity} is an entanglement
monotone, since it obeys the following properties. (i) It vanishes on all
biseparable states. (ii) The quantity does not increase under protocols that
consist of local operations of each party and classical communication between
them. 
(iii) It is convex.
In addition, in the bipartite case, $\mathcal{N}(\rho)$ becomes the negativity \cite{VidalComputablemeasureentanglement2002,
 PlenioLogarithmicNegativityFull2005}.
Note that the GME negativity defined here is slightly different from the GME
monotone defined in Ref. \cite{JungnitschTamingMultiparticleEntanglement2011}. 
The GME negativity defined here is a natural choice as the convex roof of bipartite negativity in Eq.~\eqref{eq:GME_negativity}.
\begin{proof}
Property (i) is direct, since any biseparable state $\rho$ admits the
decomposition $\rho=\sum_S p_S \rho_S$, where $\rho_S$ is separable with respect
to the partition $S|S^c$. Therefore, $\mathcal{N}_S(\rho_S)=0$ and
thus $\mathcal{N}(\rho)=0$.

Property (ii) follows from that $\mathcal{N}_S$ does not increase under local
operations of $S$ and $S^c$ and classical communication between them, which
are the superset of local operations of each party and classical communication
between them.

To show the convexity, let $\mathcal{N}(\rho_i)=\sum_S
p_S^{(i)} \mathcal{N}_S(\rho_S^{(i)})$ be the negativity of $\rho_i$ with $\rho_i= \sum_S p_S^{(i)} \rho_S^{(i)}$.
Then we have
\begin{equation}
    \begin{aligned}
        \sum_i p_i \mathcal{N}(\rho_i)
        = &\sum_S \tilde{p}_S \sum_i \tilde{p}_S^{(i)} \mathcal{N}_S\qty(\rho_S^{(i)})\\
        \ge &\sum_S \tilde{p}_S \mathcal{N}_S\qty(\sum_i \tilde{p}_S^{(i)} \rho_S^{(i)})\\
        \ge & \mathcal{N}\qty(\sum_i p_i\rho_i),
    \end{aligned}
\end{equation}
where $\tilde{p}_S=\sum_i p_i p_S^{(i)}$ and $\tilde{p}_S^{(i)}=p_i p_S^{(i)}/\tilde{p}_S$.
The first inequality comes from the convexity of $\mathcal{N}_S$, and the second inequality holds since $\sum_S\tilde{p}_S\sum_i\tilde{p}_S^{(i)}\rho_S^{(i)}$ is a decomposition of $\sum_i p_i \rho_i$.
\end{proof}

\textit{Proof of Eq.~\eqref{eq:GME_negativity_dual}.}---%
First, one can easily verify that the bipartite negativity
$\mathcal{N}_S(\rho_S)$ can be obtained by
\begin{equation}
	\begin{aligned}
        \min_{X_S}:  \quad & \Tr(X_S)\\
		\text{s. t.}: \quad
        &X_S \succeq 0,\ X_S \succeq -\rho_S^{\Gamma_S}.
	\end{aligned}
\end{equation}
Thus, the GME negativity $\mathcal{N}(\rho)$ can also be obtained by solving the
following optimization problem:
\begin{equation}
	\begin{aligned}
		\min_{p_S,\rho_S, X_S}: \quad & \sum_S p_S \Tr X_S\\
		\text{s. t.}: \quad & \rho=\sum_Sp_S\rho_S,\\
        \quad & p_S\ge 0,~\rho_S\succeq 0,~\Tr(\rho_S)=1,\\
        \quad & X_S\succeq 0,~X_S\succeq -\rho_S^{\Gamma_S},~\forall S.
	\end{aligned}
    \label{eqa:GMENegOpt}
\end{equation}
Let $\tilde{\rho}_S=p_S\rho_S$ and $\tilde{X}_S=p_SX_S$.
Eq.~\eqref{eqa:GMENegOpt} can be cast as the following SDP:
\begin{equation}
	\begin{aligned}
		\min_{\tilde{\rho}_S, \tilde{X}_S}: \quad & \sum_S \Tr \tilde{X}_S\\
		\text{s. t.}: \quad & \rho=\sum_S \tilde{\rho}_S,
        ~\tilde{\rho}_S\succeq 0,\\
        \quad &\tilde{X}_S\succeq 0,
        ~\tilde{X}_S\succeq-\tilde{\rho}_S^{\Gamma_S},~\forall S.
	\end{aligned}
\end{equation}
whose dual problem is exactly Eq.~\eqref{eq:GME_negativity_dual}, i.e.,
\begin{equation}\label{eqa:GMENegDual}
	\begin{aligned}
		\max_{W,Q_S}: \quad& -\Tr(W\rho)\\
		\text{s. t.}: \quad
		& 0 \preceq Q_S \preceq \mathbb{I},~
        Q_S^{\Gamma_S}\preceq W,~\forall S,\\
	\end{aligned}
\end{equation}
and moreover, the strong duality holds due to Slater's condition \cite{VandenbergheSemidefiniteProgramming1996}.

\section{Proof of Observation~\ref{obs:lp}}\label{app:lp}

\begin{lem}
	For an SSW $W$, if $W \succeq Q_S^{\Gamma_S}$ for some $Q_S \succeq 0$,
    then there exists $\tilde{Q}_S \succeq 0$ in the form of
	\begin{equation}\label{eq:q_gamma}
		\tilde{Q}_S^{\Gamma_S}=\sum_{T \subseteq E} q_T(S) \bigotimes_{e \in T} \Phi_e
		\otimes \bigotimes_{e \notin T} (\mathbb{I}_e-\Phi_e),
	\end{equation} 
	such that $W \succeq \tilde{Q}_S^{\Gamma_S}$ and $\tilde{Q}_S \succeq 0$.
	Moreover, the constraint $\tilde{Q}_S\preceq\mathbb{I}$ is also satisfied
    whenever $Q_S \preceq \mathbb{I}$ holds.
    \label{lem:twirling}
\end{lem}
\begin{proof}
	Let $\mathcal{T}$ be the linear operator on $X \in \mathcal{L}(\mathbb{C}^d \otimes \mathbb{C}^d)$, which is defined as
	\begin{equation}\label{eq:twirling}
		\begin{aligned}
			\mathcal{T}(X)=&\frac{\Tr(\Phi X)}{\Tr\Phi} \Phi + \frac{\Tr((\mathbb{I}-\Phi)X)}{\Tr(\mathbb{I}-\Phi)} (\mathbb{I}-\Phi)\\
			=&\Tr(\Phi X) \Phi + \Tr(\frac{\mathbb{I}-\Phi}{d^2-1}X) (\mathbb{I}-\Phi).
		\end{aligned}
	\end{equation}
	In fact, $\mathcal{T}$ is equivalent to the $U \otimes U^*$-twirling operation defined in \cite{TerhalSchmidtnumberdensity2000}.
	The twirling operation $\mathcal{T}$ has the following properties:
	\begin{itemize}
		\item[(i)] It takes any operator into the subspace $\mathrm{span}\qty{\Phi,
		\mathbb{I}-\Phi}$.
		\item[(ii)] It keeps invariant operators in $\mathrm{span}\qty{\Phi,
		\mathbb{I}-\Phi}$, i.e.,
		$\mathcal{T}(a\Phi+b(\mathbb{I}-\Phi))=a\Phi+b(\mathbb{I}-\Phi)$.
		\item[(iii)] If $a\mathbb{I}\preceq X \preceq b\mathbb{I}$,
		then $a\mathbb{I}\preceq \mathcal{T}(X)
		\preceq b\mathbb{I}$.
	\end{itemize}
	Similarly, let $\mathcal{T}'$ be the linear operator defined as
	\begin{equation}
		\mathcal{T}'(X) = \frac{\Tr(P^+ X)}{\Tr(P^+)} P_+ + \frac{\Tr(P^- X)}{\Tr(P^-)} P^-,
	\end{equation}
	where $P^\pm=\frac{\mathbb{I}\pm V}{2}$ are orthogonal projectors with
	$V=\sum_{i,j=0}^{d-1} \ketbra{i}{j} \otimes \ketbra{j}{i}$ being the swap
	operator. $\mathcal{T}'$ is equivalent to the $U \otimes U$-twirling operation
	in \cite{WernerQuantumstatesEinstein1989} and also have the similar properties
	as $\mathcal{T}$, with $\mathrm{span}\qty{\Phi, \mathbb{I}-\Phi}$ replaced by
	$\mathrm{span}\qty{P^+, P^-}$.
    Moreover, we have
    \begin{equation}\label{eq:twirling_dual}
        \mathcal{T}(X^\Gamma) = \qty[\mathcal{T}'(X)]^\Gamma,
    \end{equation}
    where $\Gamma$ denotes the partial transpose (with respect to one party).
	
	Let $\mathcal{T}_e$ and $\mathcal{T}'_e$ denote the operations $\mathcal{T}$
	and $\mathcal{T}'$ on $e \in E$, respectively. For $W \succeq Q_S^{\Gamma_S}$,
	we have $(\bigotimes_{e \in E} \mathcal{T}_e) (W - Q_S^{\Gamma_S}) \succeq 0$
	from property (iii)
    (property (iii) also holds for the multi-party case),
    and property (ii) then leads to that
	\begin{equation}
		W = \qty(\bigotimes_{e \in E} \mathcal{T}_e) (W) \succeq \qty(\bigotimes_{e \in E} \mathcal{T}_e) (Q_S^{\Gamma_S}) \equiv \tilde{Q}_S^{\Gamma_S}.
	\end{equation}
	From property (i), $\tilde{Q}_S^{\Gamma_S}$ is in the form of Eq.~\eqref{eq:q_gamma}. Then by Eq.~\eqref{eq:twirling_dual}, we have
	\begin{equation}\label{eq:q_twirling}
		\tilde{Q}_S = \qty(\bigotimes_{e \notin T(S)} \mathcal{T}_e \otimes \bigotimes_{e \in T(S)} \mathcal{T}'_e) (Q_S),
	\end{equation}
	then $\tilde{Q}_S \succeq 0$ ($\tilde{Q}_S \preceq \mathbb{I}$) if $Q_S \succeq 0$ ($Q_S \preceq \mathbb{I}$) from property (iii). This completes the proof.
\end{proof}

\textit{Proof of Observation~\ref{obs:lp}.}---% 
An observable $W$ is a fully decomposable witness if and only if there exists $Q_S \succeq 0$ such that $W \succeq Q_S^{\Gamma_S}$ for all subsystems $S$ \cite{JungnitschTamingMultiparticleEntanglement2011}.
For an SSW $W$ in the form of Eq.~\eqref{eq:SSW}, Lemma~\ref{lem:twirling}
implies that we only need to consider $Q_S^{\Gamma_S}$ in the form of
Eq.~\eqref{eq:q_gamma}, i.e.,
\begin{equation}
    Q_S^{\Gamma_S}=\sum_{T \subseteq E} q_T(S) \bigotimes_{e \in T} \Phi_e
    \otimes \bigotimes_{e \notin T} (\mathbb{I}_e-\Phi_e),
\end{equation} 
then the condition $W\succeq Q_S^{\Gamma_S}$ is equivalent to $w_T \ge q_T(S)$ for all $T\subseteq E$, as $\bigotimes_{e \in T} \Phi_e \otimes \bigotimes_{e \notin T} (\mathbb{I}_e-\Phi_e)$ are orthogonal projetors for different $T$.

Note that for a subsystem $S$,
\begin{equation}
    \qty{T\mid T\subseteq E} = \qty{T\cup\tilde{T} \mid T\subseteq T(S), \tilde{T} \subseteq T(S)^c},
\end{equation}
where $T(S)^c=E\setminus T(S)$, then $Q_S^{\Gamma_S}$ can be reformulated as
\begin{equation}
    Q_S^{\Gamma_S} = \sum_{\substack{T\subseteq T(S),\\
    \tilde{T} \subseteq T(S)^c}} q_{T\cup\tilde{T}}(S) \bigotimes_{e\in T\cup\tilde{T}}\Phi_e \otimes \bigotimes_{e\in E\setminus (T\cup\tilde{T})}(\mathbb{I}_e-\Phi_e).
\end{equation}
Since $\Phi^\Gamma=\frac{P^+-P^-}{d}$ and $(\mathbb{I}-\Phi)^\Gamma=\frac{(d-1)P^++(d+1)P^-}{d}$, we have
\begin{equation}
    Q_S=
    \sum_{\substack{T\subseteq T(S),\\
    \tilde{T} \subseteq T(S)^c}}
    \tfrac{q_{T\cup\tilde{T}}(S)}{d^\abs{T(S)}}  A_T(S) \otimes \bigotimes_{e\in\tilde{T}} \Phi_e \otimes \bigotimes_{e\in T(S)^c\setminus\tilde{T}} (\mathbb{I}_e-\Phi_e),
\end{equation}
where
\begin{equation}
    \begin{aligned}
        A_T(S)&=\bigotimes_{e\in T} (P_e^+-P_e^-)\\
        \otimes &\ \bigotimes_{e \in T(S)\setminus T} ((d-1)P_e^+ + (d+1)P_e^-).
    \end{aligned}
\end{equation}
The set of operators $\qty{A_T(S)\mid T\subseteq T(S)}$ and the set of operators $\qty{B_T(S)\mid T\subseteq T(S)}$ where $B_T(S)=\bigotimes_{e\in T} P_e^+ \otimes \bigotimes_{e\in T(S)\setminus T} P_e^-$ span the same vector space.
Then $A_T(S)$ can be expanded by $\qty{B_T(S)\mid T\subseteq T(S)} $,
and it is easy to verify that the coefficients are given by
\begin{equation}
    \mathbold{x}_T(S) = \bigotimes_{e\in T} (1, -1) \otimes \bigotimes_{e\in T(S)\setminus T} (d-1, d+1).
\end{equation}
Therefore, $Q_S$ can be reformulated as
\begin{equation}\label{eq:qsbt}
    Q_S=\sum_{\substack{T\subseteq T(S),\\
    \tilde{T} \subseteq T(S)^c}} \tfrac{r_{T\cup\tilde{T}}(S)}{d^\abs{T(S)}} B_T(S) \otimes \bigotimes_{e\in\tilde{T}} \Phi_e \otimes \bigotimes_{T(S)^c\setminus\tilde{T}} (\mathbb{I}_e-\Phi_e),
\end{equation}
where $r_{T\cup \tilde{T}}(S)$ are given by $\sum_{T\subseteq T(S)} q_{T\cup\tilde{T}}(S) \mathbold{x}_T(S)$ and the operators within the summation are orthogonal projectors for different $T\cup \tilde{T}$.
Then the condition $Q_S \succeq 0$ is equivalent to $r_{T\cup\tilde{T}}(S) \ge 0$ for all $T\cup\tilde{T}$, i.e., 
\begin{equation}
    \sum_{T\subseteq T(S)} q_{T\cup\tilde{T}}(S) \mathbold{x}_T(S) \ge 0, \quad \forall \tilde{T}\subseteq E\setminus T(S),
\end{equation}
which completes the proof.

In addition, the condition $0\preceq Q_S \preceq \mathbb{I}$ is equivalent to $0\le \frac{r_{T\cup\tilde{T}}(S)}{d^\abs{T(S)}}\le 1$, i.e., 
\begin{equation}
    0 \le \sum_{T\subseteq T(S)} q_{T\cup\tilde{T}}(S) \mathbold{x}_T(S) \le d^\abs{T(S)}, \quad \forall \tilde{T}\subseteq E\setminus T(S).
\end{equation}

\bigskip
\textit{Proof of Eq.~\eqref{eq:negativity}.}---%
Let $\mathcal{S}$ denote the set of feasible witnesses in Eq.~\eqref{eq:GME_negativity_dual}, and $\mathcal{W}$ denote the set of operators in the form of Eq.~\eqref{eq:SSW}.
If we replace the condition $W\in \cal{S}$ in Eq.~\eqref{eq:GME_negativity_dual} by $W \in \mathcal{S} \cap \mathcal{W}$, we can get a lower bound of GME negativity.
Then the SDP
\begin{equation}
	\begin{aligned}
		\max_{W,Q_S}: \quad& -\Tr(W\rho)\\
		\text{s. t.}: \quad
		& 0 \preceq Q_S \preceq \mathbb{I},~
        Q_S^{\Gamma_S}\preceq W,~\forall S
	\end{aligned}
\end{equation}
in Eq.~\eqref{eq:GME_negativity_dual} becomes
\begin{equation}\label{eq:sdp_restrcited}
	\begin{aligned}
		\max_{W,Q_S}: \quad& -\Tr(W\rho)\\
		\text{s. t.}: \quad
		& 0 \preceq Q_S \preceq \mathbb{I},~
        Q_S^{\Gamma_S}\preceq W,~\forall S,\\
        & W = \sum_{T\subseteq E}w_T\textstyle\bigotimes_{e\in T}\Phi_e\otimes \bigotimes_{e\notin T}(\mathbb{I}_e-\Phi_e).
	\end{aligned}
\end{equation}
On the one hand, $\Tr(W\rho)=\sum_{T\subseteq E} w_T p_T$, where $p_T=\Tr(\rho P_T)$ with $P_T=\bigotimes_{e\in T}\Phi_e\otimes\bigotimes_{e\notin T} (\mathbb{I}_e-\Phi_e)$. In the case of $\rho=\bigotimes_{e\in E}\rho_e$, $p_T$ can be further simplified to $p_T=\prod_{e\in T}\Tr(\rho_e\Phi_e)
\prod_{e\notin T}(1-\Tr(\rho_e\Phi_e))$.
On the other hand, from the proof of Observation~\ref{obs:lp}, the condition $0\preceq Q_S\preceq \mathbb{I}$ and $Q_S^{\Gamma_S}\preceq W$ are equivalent to that there exist $q_T(S)\le w_T$ such that
\begin{equation}
    0 \le \sum_{T\subseteq T(S)} q_{T\cup\tilde{T}}(S) \mathbold{x}_T(S) \le d^\abs{T(S)}, \quad \forall \tilde{T}\subseteq E\setminus T(S).
\end{equation}
Therefore, the SDP in Eq.~\eqref{eq:sdp_restrcited} reduces to the LP
\begin{equation}
	\begin{aligned}
		\max_{w_T,q_T}: \quad& -\sum_{T \subseteq E}w_Tp_T\\
		\text{s. t.}: \quad & 
		q_T(S) \le w_T,\\
		&0\le\sum_{T\subseteq T(S)} q_{T\cup\tilde{T}}(S) \vect{x}_T(S)\le d^{\abs{T(S)}}
	\end{aligned}
\end{equation}
in Eq.~\eqref{eq:negativity}, which provides a lower bound of GME negativity for arbitrary $\rho$.

Moreover, for $\rho\in\mathrm{span}\qty{P_T\mid T\subseteq E}$ (e.g., the state distribution process is subject to white noise), the bound is tight.
In fact, for $\rho\in\mathrm{span}\qty{P_T\mid T\subseteq E}$, if $W\in \cal{S}$, we can construct an SSW $\tilde{W}\in \mathcal{S}\cap \mathcal{W}$ such that $\Tr(W\rho)=\Tr(\tilde{W}\rho)$, thus the bound is tight.
Let $\tilde{W}=(\bigotimes_{e\in E} \mathcal{T}_e)(W)$. It is easy to see that $\tilde{W}\in\mathcal{W}$ from property (i) of $\mathcal{T}$ and that $\Tr(W\rho)=\Tr(\tilde{W}\rho)$ as the operation $\bigotimes_{e \in E}\mathcal{T}_e$ keeps invariant the overlap with $P_T$, that is, $\Tr(P_T [\qty(\bigotimes_{e \in E}\mathcal{T}_e)(X)])=\Tr(P_T X)$ for any $T \subseteq E$ and any $X$.
To see that $\tilde{W}\in\mathcal{S}$, recall that $W\in\mathcal{S}$, hence there exists $0\preceq Q_S\preceq \mathbb{I}$ such that $W\succeq Q_S^{\Gamma_S}$. From the property (iii) of $\mathcal{T}$, we have $\qty(\bigotimes_{e \in E}\mathcal{T})(W-Q_S^{\Gamma_S}) \succeq 0$, and thus
\begin{equation}
    \tilde{W} = \qty(\bigotimes_{e \in E}\mathcal{T}_e)(W) \succeq \qty(\bigotimes_{e \in E}\mathcal{T}_e) \qty(Q_S^{\Gamma_S})=\tilde{Q}_S^{\Gamma_S},
\end{equation}
where $0 \preceq \tilde{Q}_S \preceq \mathbb{I}$ from Eq.~\eqref{eq:q_twirling} and property (iii) of $\mathcal{T}$ and $\mathcal{T}'$.
Therefore, $\tilde{W}\in \mathcal{S}$.

\section{Network topology and the cut space}\label{app:property}
Consider a graph $G=(V, E)$, where $v \in V$ denote the vertices and unordered pairs $e=(u, v) \in E$ with $u \ne v$ for $u, v \in V$ denote the edges.

The ends of an edge are said to be \textit{incident} with the edge, and vice versa, that is, a vertex $v\in V$ is incident with an edge $(u, w)\in E$ if and only $u=v$ or $w=v$. 
Two vertices which are incident with a common edge are \textit{adjacent}, i.e., $u, v\in V$ are adjacent if and only if $(u, v) \in E$.
The \textit{degree} $\deg(v)$ of a vertex $v\in V$ is the number of edges incident with $v$, or the number of vertices adjacent to $v$.

A \textit{walk} in $G=(V, E)$ is a finite non-null sequence $W=v_0 e_1 v_1 e_2 v_2\dots e_k v_k$, whose terms are alternatively vertices and edges, such that $e=(v_{i-1}, v_i)$ for $e \in E$ and $1 \le i \le k$.
A walk $W=v_0 e_1 v_1 e_2 v_2\dots e_k v_k$ is closed if $k \ne 0$ and $v_0 = v_k$.
A walk $W=v_0 e_1 v_1 e_2 v_2\dots e_k v_k$ is a $v_0-v_k$ \textit{path} if $v_i \ne v_j$ for all $i \ne j$.
If for any $u, v \in V$, there exists a $u-v$ path in $G$, we call $G$ \textit{connected}.
Here, we only consider connected graphs, since networks characterized by unconnected graphs can be partitioned into two separate parts, thus GME does not exist even if the shared states are all perfect $\Phi$.

The subsets of $E$ in $G$ form a vector space $\mathcal{E}(G)$ called \textit{edge space} over $\mathbb{F}_2=\qty{0, 1}$, with the addition being the symmetric difference
\begin{equation}
	T \oplus T' = \qty{e \in E \mid e \in T, e \notin T' \text{ or } e \notin T, e \in T'}
\end{equation}
and the scalar multiplication
\begin{equation}
	0 \cdot T = \emptyset,\quad 1 \cdot T = T.
\end{equation}
A set $K$ of edges is a \textit{cut} in $G=(V, E)$ if there exists a bipartition $S|S^c$ of $V$ such that $K=T(S)$ (recall that $T(S)=\qty{(u,v) \in E \mid u\in S, v\notin S \text{ or } u\notin S, v\in S}$ in Eq.~\eqref{eq:defTS}), which exactly captures the set of partially transposed edges when we perform partial transpose with respect to the bipartition $S|S^c$.
For example, the set $\{e_1, e_4\}$ in Fig.~\ref{fig:graph} is the cut for the bipartition $u_1u_2u_3|u_4u_5$.
Together with $\emptyset$, all cuts $K$ in $G$ form the \textit{cut space} $\mathcal{B}(G)$ as a subspace of $\mathcal{E}(G)$. 

\begin{figure}
    \centering
    \includegraphics[width=0.25\textwidth]{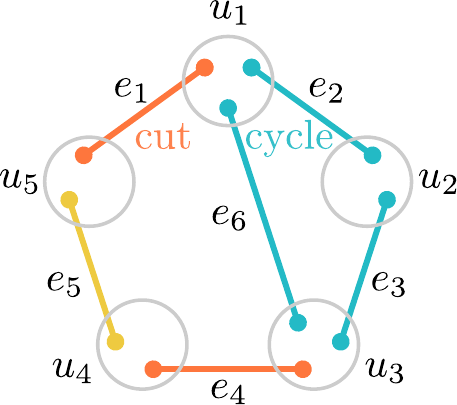}
    \caption{An example of a graph, where the blue walk $u_1e_2u_2e_3u_3e_6u_1$ characterizes a cycle, and the set $\{e_1, e_4\}$ of the orange edges is the cut for the bipartition $u_1u_2u_3|u_4u_5$.}
    \label{fig:graph}
\end{figure}

A subgraph $C$ of $G$ is a \textit{cycle} if it consists of all vertices and edges of a closed walk $W=v_0 e_1 v_1 e_2 v_2\dots e_k v_k$ where $v_0=v_k$ and $v_i \ne v_j$ for all $i\ne j$ and $0 \le i,j \le k-1$, that is, $C=(\qty{v_0, v_1, \dots, v_{k-1}}, \qty{e_1, e_2, \dots, e_k})$.
For example, there are three cycles in Fig.~\ref{fig:graph}, which are characterized by $u_1e_2u_2e_3u_3e_6u_1$, $u_1e_2u_2e_3u_3e_4u_4e_5u_5e_1u_1$, and $u_1e_6u_3e_4u_4e_5u_5e_1u_1$.
The \textit{cycle space} $\mathcal{C}(G)$ is the subspace of $\mathcal{E}(G)$ spanned all the cycles of $G$, or more exactly, the sets of edges of the cycles.
The dimension of $\mathcal{C}(G)$ is called the \textit{cyclomatic number}.

It is known that for a graph $G$, $\mathcal{B}(G)$ is the orthogonal complement of $\mathcal{C}(G)$, that is, $\mathcal{B}(G)=\mathcal{C}(G)^\perp$ \cite{DiestelGraphtheory2025}.
Based on this, we can find $\mathcal{B}(G)$ easily for some graphs.
For instance, a tree graph is a connected graph without any cycle, see examples in Fig.~\ref{fig:networks}(a), (b) and (c).
Therefore, we have $\mathcal{C}(G)=\emptyset$ and thus $\mathcal{B}(G)=\mathcal{E}(G)$, the collection of all subsets of $E$.
Another example is the ring network consisting of a single cycle, see examples in Fig.~\ref{fig:networks}(d) and (e).
Hence $\mathcal{C}(G)=\qty{\emptyset, E}$, then it is direct to verify that $\mathcal{B}(G)$ is composed of all subsets of $E$ of even elements.

Without loss of generality, we can label the $N$ edges in $E$ by $1$ to $N$.
Denoting by $\mathcal{G}_N$ the set of all graphs with edge set $E=\qty{1, 2, \dots, N}$, we define a partial order on $\mathcal{G}_N$ via the inclusion of the cut space $\mathcal{B}(G)$ for $G \in \mathcal{G}_N$.
That is, we say that $G_1 \prec G_2$ whenever the cut space of
$G_1$ is contained in that of $G_2$, i.e., $\mathcal{B}(G_1) \subseteq
\mathcal{B}(G_2)$. This, in turn, gives rise to an equivalence relation:
two graphs $G_1,G_2\in\mathcal{G}_N$ are considered equivalent, written $G_1
\sim G_2$ whenever $\mathcal{B}(G_1) = \mathcal{B}(G_2)$.
For example, denote the graphs in Fig.~\ref{fig:differentTopology}(a), (d), and (e) by $G_1$, $G_2$, and $G_3$, respectively.
All of them have three edges $\qty{1, 2, 3}$, that is, $G_1, G_2, G_3 \in \mathcal{G}_3$. Then we have
\begin{equation}
    \begin{aligned}
        & \qty{\emptyset, \qty{1, 2}, \qty{2, 3}, \qty{1, 3}} = \mathcal{B} (G_3) \subsetneq \mathcal{B}(G_1)=\mathcal{B}(G_2)\\
        &=\qty{\emptyset, \qty{1}, \qty{2}, \qty{3}, \qty{1, 2}, \qty{1, 3}, \qty{2, 3}, \qty{1, 2, 3}}.
    \end{aligned}
\end{equation}
Therefore, we have $G_3 \prec G_1 \sim G_2$.

\textit{Proof of Observation~\ref{obs:order}.}---%
Let $W$ be an SSW for $G_2$, then it satisfies Eq.~\eqref{eq:linearConstraintsK} for all $K \in \mathcal{B}(G_2)$.
Since $G_1 \prec G_2$, $W$ also satisfies Eq.~\eqref{eq:linearConstraintsK} for all $K \in \mathcal{B}(G_1) \subseteq \mathcal{B}(G_2)$, thus is an SSW for $G_1$, which proves (a).
Then (b) follows since $G_1 \sim G_2$ if and only if $G_1 \prec G_2$ and $G_2 \prec G_1$.

\textit{Proof of Observation~\ref{obs:tree}.}---%
As proved above, the cut space $\mathcal{B}(G_\mathrm{tree})$ of a tree graph $G_\mathrm{tree} \in \mathcal{G}_N$ consists of all subset of $E=\qty{1, 2, \dots, N}$.
Therefore, for any $G \in \mathcal{G}_N$, $G \prec G_\mathrm{tree}$.
Then Observation~\ref{obs:tree} follows from Observation~\ref{obs:order}.

\section{$E$-SSWs and $\mathcal{B}(G)$-SSWs}\label{app:SSW_family}
\textit{$E$-SSWs.}---%
The $E$-SSW $W_E$ is defined as
\begin{equation}\label{eq:E-SSW2}
	\begin{aligned}
		&W_E = -d\bigotimes_{e \in E} \Phi_e + \frac{d+1}{2} \bigotimes_{e \in E} \frac{\mathbb{I}_e+d\Phi_e}{d+1}\\
		&= -\frac{d-1}{2} \bigotimes_{e \in E} \Phi_e + \frac{d+1}{2} \sum_{\emptyset\ne T \subseteq E} \bigotimes_{e \in T} \frac{\mathbb{I}_e-\Phi_e}{d+1} \otimes \bigotimes_{e \notin T} \Phi_e.
	\end{aligned}
\end{equation}
In the following, we prove that $W_E$ is an SSW for tree networks.
Observation~\ref{obs:tree} implies that $W_E$ is also an SSW for any other network.

Recall from Appendix~\ref{app:property} that, for tree networks, the cut space contains all subsets $T\subseteq E$.
For any nonempty $T \subseteq E$, we have
\begin{equation}
	W_E \succeq Q_T^{\Gamma_T} \otimes \bigotimes_{e \notin T} \Phi_e,
\end{equation}
where 
\begin{equation}
	Q_T^{\Gamma_T} = \frac{1}{2\abs{T}} \sum_{e \in T} \qty[(\mathbb{I}_e-d\Phi_e) \otimes \bigotimes_{e' \in T \setminus \qty{e}} \frac{\mathbb{I}_{e'}+d\Phi_{e'}}{d+1}]
\end{equation}
and $\Gamma_T$ denotes the partial transpose of all edges $e \in T$ without ambiguity.
Therefore,
\begin{equation}
Q_T=\frac{1}{\abs{T}}\sum_{e \in T} P_e^- \otimes \bigotimes_{e' \in T
\setminus \qty{e}} \frac{2P_{e'}^+}{d+1},    
\end{equation}
which satisfies $0 \preceq Q_T \preceq \bigotimes_{e\in T}\mathbb{I}_e$.
It follows that $0\preceq Q_T\otimes\bigotimes_{e\notin T} \Phi_e \preceq \mathbb{I}$.

To see $W_E \succeq Q_T^{\Gamma_T} \otimes \bigotimes_{e \notin T} \Phi_e$, we expand $Q_T^{\Gamma_T}$ with positive semidefinite operators $\bigotimes_{e\in T'} \frac{\mathbb{I}_e-\Phi_e}{d+1}\otimes\bigotimes_{e\in T\setminus T'} \Phi_e$ as Eq.~\eqref{eq:E-SSW2}, which are orthogonal with each other for different $T'$. Then
\begin{widetext}
\begin{equation}
    \begin{aligned}
        Q_T^{\Gamma_T} = &\frac{1}{2\abs{T}} \sum_{e\in T} \Bigg[\qty((d+1)\frac{\mathbb{I}_e-\Phi_e}{d+1}-(d-1)\Phi_e) \otimes \bigotimes _{e'\in T\setminus\qty{e}} \qty(\frac{\mathbb{I}_{e'}-\Phi_{e'}}{d+1}+\Phi_{e'})\Bigg]\\
        = &-\frac{d-1}{2}\bigotimes_{e\in T} \Phi_e + \sum_{\substack{T' \subseteq T\\ T' \ne \emptyset}} a(T, T') \qty(\bigotimes_{e \in T'} \frac{\mathbb{I}_e-\Phi_e}{d+1} \otimes \bigotimes_{e \in T \setminus T'} \Phi_e),
    \end{aligned}
\end{equation}
\end{widetext}
where $a(T, T')=\frac{1}{2\abs{T}}\sum_{e\in T} c(e, T')$ and
\begin{equation}
    c(e, T')=
    \begin{cases}
        d+1,  &e \in T',\\
        -(d-1), &e \in T\setminus T'.
    \end{cases}
\end{equation}
Therefore,
\begin{equation}
    a(T, T') = \frac{\abs{T'}(d+1)-(\abs{T}-\abs{T'})(d-1)}{2\abs{T}}  \le
    \frac{d+1}{2},
\end{equation}
and the equality saturates when $T'=T$.
Then we have $W_E \succeq Q_T^{\Gamma_T}\otimes\bigotimes_{e\notin T}\Phi_e$.
This shows that $W_E$ is a legitimate SSW. Moreover, it satisfies the constraints in Eq.~\eqref{eq:GME_negativity_dual}.

\textit{$\mathcal{B}(G)$-SSWs.}---%
The $\mathcal{B}(G)$-SSW $W_{\mathcal{B}(G)}$ is defined as
\begin{equation}
	W_{\mathcal{B}(G)} = -\bigotimes_{e\in E}\Phi_e
    + \sum_{\emptyset \ne K \in \mathcal{B}(G)}
    \bigotimes_{e\in K}\frac{\mathbb{I}_e-\Phi_e}{d-1}
    \otimes \bigotimes_{e\notin K}\Phi_e.
\end{equation}

Note that $W_{\mathcal{B}(G)}\succeq Q_K^{\Gamma_K}$, where
\begin{equation}
    Q_K^{\Gamma_K}=\qty(-\bigotimes_{e\in K} \Phi_e + \bigotimes_{e\in K} \frac{\mathbb{I}_e-\Phi_e}{d-1}) \otimes \bigotimes_{e\notin K} \Phi_e,
\end{equation}
and
\begin{equation}
    \begin{aligned}
        &Q_K\\
        =&\qty[-e\bigotimes_{e\in K} \qty(\tfrac{P_e^+}{d}-\tfrac{P_e^-}{d}) + \bigotimes_{e\in K}\qty(\tfrac{P_e^+}{d} + \tfrac{(d+1)P_e^-}{d(d-1)})]\otimes\bigotimes_{e\notin K} \Phi_e\\
        \succeq &0
    \end{aligned}
\end{equation}
for all nonempty cuts $K\in\mathcal{B}(G)$. 
Since $\mathcal{B}(G)$ contains the sets $T(S)$ for all subsystems $S$, this validates that $W_{\mathcal{B}(G)}$ is an SSW.

\section{SSWs for permutation-invariant states in cactus networks}\label{app:cactus}

The simplification from SDP to LP substantially reduces the complexity of the optimization problem and removes the dimension-dependent restriction, however, it still has variables and constraints increasing exponentially with the scale of the networks.
For symmetric states in networks with specific topology, the variables and constraints can be reduced.
More specifically, here we consider the permutation-invariant states in cactus networks.
By permutation-invariant, we mean that a state remains invariant under the exchange of any two edges of the state.

A cactus graph is a connected graph $G=(V, E)$ such that every $e \in E$ belongs to at most one cycle in $G$ \cite{WestIntroductionGraphTheory2001}.
See Fig.~\ref{fig:cactus} for examples of cactus graph.
For a cactus graph $G$, the cyclomatic number $n_c$ (dimension of cycle space $\mathcal{C}(G)$) is exactly the number of cycles in the network.
Note that tree networks and ring networks are cactus graphs with cyclomatic number $0$ and $1$, respectively.
\begin{figure}
    \centering\includegraphics[width=\linewidth]{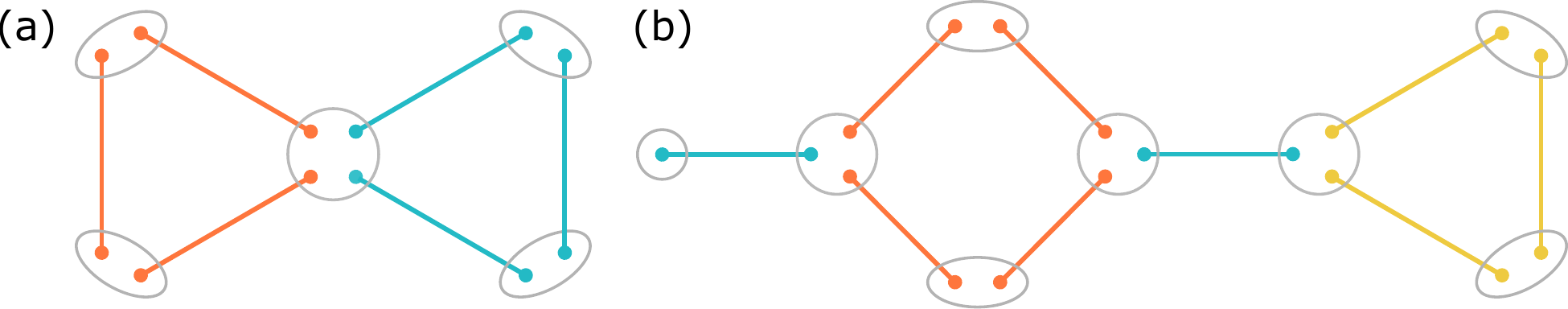}
    \caption{Examples of cactus networks. The edges of the same color are equivalent for permutation invariant states.
    }
    \label{fig:cactus}
\end{figure}

\begin{observation}
	For a permutation-invariant state $\rho$ in a cactus network with cyclomatic number $n_c$, an SSW can be determined by at most $\qty(1+\frac{\abs{E}}{n_c+1})^{n_c+1}$ coefficients.
	The corresponding LP requires at most
	$\frac{2}{12^{n_c+1}} \qty(3+\frac{\abs{E}}{n_c+1})^{3n_c+3}$
	variables, and
	$\frac{6}{12^{n_c+1}} \qty(3+\frac{\abs{E}}{n_c+1})^{3n_c+3}$
	constraints.
\end{observation}
\begin{observation}
	For permutation-invariant states in tree networks and ring networks, the problem of searching for SSWs can be cast as an LP of $O(\abs{E}^3)$ variables and constraints, and the SSW can be determined by $\abs{E}+1$ coefficients.
\end{observation}

This shows that for cactus networks with a fixed number of cycles, the scale of
the problem only depends polynomially on $\abs{E}$.
Notably, the requirement for permutation invariance can be relaxed to hold
separately in each cycle and in the set of edges not in any cycle, which
enlarges the application of the result. 

As the proof is rather lengthy, we first illustrate it with an example of the bilocal network, and then provide the full proof at the end of this section.

\textit{The example of bilocal network.}---%
For a permutation invariant state $\rho$ in the bilocal network of Fig.~\ref{fig:networks}(a), the SSW $W$ can be assumed to be the following form:
\begin{equation}
    W=w_0 \Phi^{\otimes 2} + w_1\mathcal{P}[\Phi\otimes(\mathbb{I}-\Phi)]+w_2(\mathbb{I}-\Phi)^{\otimes 2},
\end{equation}
where $X\otimes Y$ denotes $X_{AB_1}\otimes Y_{B_2C}$, and $\mathcal{P}[\Phi\otimes(\mathbb{I}-\Phi)]:=\Phi\otimes(\mathbb{I}-\Phi) + (\mathbb{I}-\Phi)\otimes\Phi$.
Then the constraints for $W$ are that $W \succeq Q_A^{\Gamma_A}$ and $W \succeq Q_B^{\Gamma_B}$, where $0 \preceq Q_A, Q_B \preceq \mathbb{I}$ and
\begin{equation}
    \begin{aligned}
        Q_A^{\Gamma_A}=&q_{0,0}^A \Phi^{\otimes 2} + q_{0,1}^A \Phi\otimes(\mathbb{I}-\Phi)\\
        &+q_{1,0}^A(\mathbb{I}-\Phi)\otimes\Phi
        +q_{1,1}^A(\mathbb{I}-\Phi)^{\otimes 2},\\
        Q_B^{\Gamma_B}=&q_0^B \Phi^{\otimes 2} + q_1^B\mathcal{P}[\Phi\otimes(\mathbb{I}-\Phi)]+q_2^B(\mathbb{I}-\Phi)^{\otimes 2}.
    \end{aligned}
\end{equation}
The constraint $W \succeq Q_C^{\Gamma_C}$ is equivalent to $W \succeq Q_A^{\Gamma_A}$ from the symmetry of the network.
Note that $Q_A^{\Gamma_A}$ can not be written in the symmetric form as only one edge is partially transposed, making it not symmetric.
Then $W \succeq Q_{A/B}^{\Gamma_{A/B}}$ is equivalent to $w_{i+j}\ge q_{i,j}^A$ and $w_i\ge q_i^B$.
We also have
\begin{equation}
	\begin{aligned}
		&Q_A = (r_0^A P^+ + r_1^A P^-)\otimes\Phi + (r_2^A P^+ + r_3^A P^-)\otimes(\mathbb{I}-\Phi),\\
		&Q_B = r_0^B {P^+}^{\otimes 2} + r_1^B (P^+ \otimes P^- + P^- \otimes P^+) + r_2^B {P^-}^{\otimes 2},
	\end{aligned}
\end{equation}
where 
\begin{equation}
	\begin{bmatrix}
		r_0^A\\[.8ex]
        r_1^A\\[.8ex]
        r_2^A\\[.8ex]
        r_3^A
	\end{bmatrix}
	=\frac{1}{d}
	\begin{bmatrix}
		q_{0,0}^A + (d-1) q_{1,0}^A\\[.8ex]
        -q_{0,0}^A+(d+1)q_{1,0}^A\\[.8ex]
        q_{0,1}^A+(d-1)q_{1,1}^A\\[.8ex]
        -q_{0,1}^A+(d+1)q_{1,1}^A
	\end{bmatrix}
\end{equation}
and
\begin{equation}
		\begin{bmatrix}
		r_0^B\\[.8ex]
        r_1^B\\[.8ex]
        r_2^B
	\end{bmatrix}
	=\frac{1}{d^2}
	\begin{bmatrix}
		q_0^B+2(d-1)q_1^B+(d-1)^2q_2^B\\[.8ex]
        -q_0^B+2q_1^B+(d^2-1)q_2^B\\[.8ex]
        q_0^B-2(d+1)q_1^B+(d+1)^2 q_2^B
	\end{bmatrix}.
\end{equation}
Then the constraints $0 \preceq Q_A, Q_B \preceq \mathbb{I}$ are equivalent to $0 \le r_i^A, r_j^B \le 1$.
Therefore, the SSW is determined by $3$ variables $w_0$, $w_1$, and $w_2$. The corresponding LP has $10$ variables $w_i$, $q_{i,j}^A$, and $q_i^B$, with $21$ constraints given by inequalities $w_{i+j}\ge q_{i,j}^A$, $w_i\ge q_i^B$, $r_i^A, r_j^B \ge 0$, and $r_i^A, r_j^B \le 1$.

The detailed proof is in the following.

\begin{proof}
	In a cactus graph, an edge belongs to at most one cycle, therefore we can classify edges with respect to which cycle they belong to. 
	Denote the set of edges in different cycles by $C_i$ with $\abs{C_i}:=c_i$ for $i=1,\dots,n_c$, and denote the set of edges not in any cycle by $C_0=E \setminus \bigcup_{i=1}^{n_c}C_i$ and $\abs{C_0}=c_0$.
	It is not difficult to verify that
	\begin{equation}
		\mathcal{B}(G)=\qty{K \mid K=\bigcup_{i=0}^{n_c}T_i, T_i \subseteq C_i\ \forall i, \abs{T_i} \text{ is even}\ \forall i \ge 1},
	\end{equation}
	which shows that the cut space consists of the edge sets which have even edges in $C_i$ for $i \ge 1$ and have arbitrary edges in $C_0$.
	In this case, if we exchange two edges in $C_i$ and obtain a new graph $G'$, it holds that $G \sim G'$, since the exchange operation does not change the number of edges in $C_i$.
	Therefore, for a given SSW $W$, if we exchange any two edges in $C_i$, the obtained operator $W'$ is still an SSW, which follows from Observation~\ref{obs:order}.
	For a permutation-invariant state $\rho$, the SSWs $W$ and $W$' have the same expectation, that is, $\Tr(W \rho)=\Tr(W' \rho)$, hence we can restrict the SSW to be invariant under these permutations, that is,
	\begin{equation}
		W = \sum_{t'_0, \dots, t'_{n_c}} w_{t'_0, \dots, t'_{n_c}} \sum_{\qty{T \subseteq E \mid t_i(T)=t'_i\ \forall i}} P_T,
	\end{equation}
	where $t_i(T):=\abs{T\cap C_i}$ and $t'_i$ ranges from $0$ to $c_i$, and $P_T=\bigotimes_{e \in T} \Phi_e \otimes \bigotimes_{e \notin T}(\mathbb{I}_e-\Phi_e)$ as defined in Appendix~\ref{app:lp}.
	Therefore, the number of coefficients of the SSW is $\prod_{i=0}^{n_c}(c_i+1)\le \qty(1+\frac{\abs{E}}{n_c+1})^{n_c+1}$ as $\sum_{i=0}^{n_c}c_i=\abs{E}$.
	
	As Eq.~\eqref{eq:linearConstraintsK}, nonempty $K \in \mathcal{B}(G)$ induces constraints of Eq.~\eqref{eq:negativity} that there exist $q_T(K)\le w_T$ such that
    \begin{equation}\label{eq:linearConstraintsK2}
		0 \le \sum_{T\subseteq K} q_{T\cup\tilde{T}}(K) \vect{x}_T(K) \le d^\abs{K},
		\quad \forall\tilde{T}\subseteq E\setminus K,
	\end{equation}
	with  $\vect{x}_T(K)=\bigotimes_{e\in T} \qty(1,-1) \otimes
	\bigotimes_{e\in K\setminus T} \qty(d-1,d+1)$.
	The symmetry of the SSWs allows us to assume that $q_{T_1\cup\tilde{T}_1}(K)=q_{T_2\cup\tilde{T}_2}(K)$ for $T_1, T_2\subseteq K$ and $\tilde{T}_1, \tilde{T}_2\in E\setminus K$ if $t_i(T_1)=t_i(T_2)$ and $t_i(\tilde{T}_1)=t_i(\tilde{T}_2)$ for all $i$.
	Then for a given $\tilde{T}$, we only need to consider $T\subseteq K$ with different $t_i(T)$, that is, only $\prod_{i=0}^{n_c}(t_i(K)+1)$ variables $q_{T\cup\tilde{T}}(K)$ is enough, instead of $2^\abs{K}$.
	Moreover, the equivalence of the variables also reduces the componentwise inequalities given by Eq.~\eqref{eq:linearConstraintsK2} from $2^{\abs{K}+1}$ to $2\prod_{i=0}^{n_c}(t_i(K)+1)$, since $r_{T\cup\tilde{T}}(K)$ in Eq.~\eqref{eq:qsbt} also satisfy that $r_{T_1\cup\tilde{T}_1}(K)=r_{T_2\cup\tilde{T}_2}(K)$ for $T_1, T_2\subseteq K$ and $\tilde{T}_1, \tilde{T}_2\in E\setminus K$ if $t_i(T_1)=t_i(T_2)$ and $t_i(\tilde{T}_1)=t_i(\tilde{T}_2)$ for all $i$.
	Similar discussion shows the for a given $K$, we only need to consider $\tilde{T}\subseteq E \setminus K$ with different $t_i(\tilde{T})$.
	Therefore, for each nonempty $K \in \mathcal{B}(G)$, there are $\prod_{i=0}^{n_c}(t_i(E \setminus K)+1)=\prod_{i=0}^{n_c}(c_i-t_i(K)+1)$ inequivalent $\tilde{T}$, each of which leads to $2\prod_{i=0}^{n_c}(t_i(K)+1)$ inequalities.
    For the same reason, we only need to consider cuts $K$ with different $t_i(K)$.
	As a result, there are $\sum_{t_0', t_1', \dots, t_{n_c}'}\prod_{i=0}^{n_c}(c_i-t_i'+1)(t_i'+1)-\prod_{i=0}^{n_c}(c_i+1)$ variables $q_T(K)$ and $2\qty[\sum_{t_0', t_1', \dots, t_{n_c}'}\prod_{i=0}^{n_c}(c_i-t_i'+1)(t_i'+1)-\prod_{i=0}^{n_c}(c_i+1)]$ inequalities from Eq.~\eqref{eq:linearConstraintsK2}, where $t'_0$ ranges over the integers from $0$ to $c_0$, and for $i\ge1$, $t'_i$ ranges over the even integers from $0$ to $c_i$.
	The term $\prod_{i=0}^{n_c}(c_i+1)$ is subtracted to exclude $\emptyset \in \mathcal{B}(G)$.
	Note that the summation and multiplication commute in this case. Let $a_i:=\sum_{t'_i}(c_i - t'_i + 1) (t'_i +1)$, then we have
	\begin{equation}
		a_i =
		\begin{cases}
			\frac{(c_0+1)(c_0+2)(c_0+3)}{6}, \quad & i=0,\\
			\qty(\lfloor \frac{c_i}{2} \rfloor + 1) \qty(c_i+1 + \lfloor \frac{c_i}{2} \rfloor \frac{3c_i-4\lfloor \frac{c_i}{2} \rfloor -2}{3}), \quad & i \ge 1
		\end{cases}
	\end{equation}
	Together with $\prod_{i=0}^{n_c}(c_i+1)$ variables $w_{t'_0,\dots,t'_{n_c}}$, there are $\prod_{i=0}^{n_c}a_i$ variables.
	Together with constraints $q_T \le w_T$, there are $3\qty( \prod_{i=0}^{n_c}a_i - \prod_{i=0}^{n_c}(c_i+1))$
	inequalities.
	Since $a_0<\frac{(c_0+3)^3}{6}$ and $a_i < \frac{(c_i+3)^3}{12}$ for all $i\ge 1$, we have
	\begin{equation}
		\prod_{i=0}^{n_c}a_i \le \frac{2}{12^{n_c+1}} \qty(3+\frac{\abs{E}}{n_c+1})^{3n_c+3}
	\end{equation}
	and
	\begin{equation}
		3\qty( \prod_{i=0}^{n_c}a_i - \prod_{i=0}^{n_c}(c_i+1)) \le \frac{6}{12^{n_c+1}} \qty(3+\frac{\abs{E}}{n_c+1})^{3n_c+3}
	\end{equation}
for general cactus networks.

	For the tree networks, calculations shows that we have $\frac{(\abs{E}+1)(\abs{E}+2)(\abs{E}+3)}{6}=O(\abs{E}^3)$ variables and $\frac{\abs{E}(\abs{E}+1)(\abs{E}+5)}{2}=O(\abs{E}^3)$ inequalities, and the SSW has $\abs{E}+1$ different coefficients.
	Similarly, for ring networks, we also have $O(\abs{E}^3)$ variables and constraints.
\end{proof}

\section{SSWs for IPEN states in different networks}\label{app:IPEN}

\subsection{Tree networks}
Observation~\ref{obs:IPEN_tree} states that the state $\rho=\bigotimes_{e \in E} \rho_e^p$ in tree networks has non-zero GME negativity if and only if $\Tr(W_E\rho)<0$ for the SSW
\begin{equation}
    W_E=-d\bigotimes_{e\in E} \Phi_e + \frac{d+1}{2} \bigotimes_{e\in E} \frac{\mathbb{I}_e+d\Phi_e}{d+1}.
\end{equation}
Furthermore, when $\Tr(W_E\rho)<0$, the GME negativity of $\rho$ is quantified by $-\Tr(W_E\rho)$, that is, $\mathcal{N}(\rho)=\max\qty{0,-\Tr(W_E \rho)}$.

\begin{proof}
Recall from Appendix~\ref{app:SSW_family} that $W_E$ satisfies constraints in Eq.~\eqref{eq:GME_negativity_dual}, then
\begin{equation}\label{eq:tree_witness_value}
\begin{aligned}
    \mathcal{N}(\rho) & \ge -\Tr(W_E \rho) \\
    & = d\qty[\tfrac{1+(d^2-1)p}{d^2}]^\abs{E} - \tfrac{d+1}{2}
	\qty[\tfrac{1+(d-1)p}{d}]^\abs{E}
\end{aligned}
\end{equation}
when $\Tr(W_E\rho)<0$.

With $W_E$, we can verify the GME of $\rho=\bigotimes_{e \in E} \rho_e^p$ if $\Tr(W_E\rho) < 0$.
We can further show that $\rho$ is a PPT mixture if $\Tr(W_E\rho) \ge 0$.
Let
\begin{align}
\gamma&=\frac{1-p}{1+(d^2-1)p},\\
\rho_T&=\bigotimes_{e \in T} (\mathbb{I}_e-\Phi_e) + (d+1)(d-1)^{\abs{T}-1}
\bigotimes_{e \in T} \Phi_e.
\end{align}
Direct calculations show that $\rho_T\succeq 0$,
$\rho_T^{\Gamma_T} \succeq 0$, and
\begin{equation}\label{eq:IPEN_decompose}
\begin{aligned}
	\rho =&\frac{2d-(d+1)[1+(d-1)\gamma]^\abs{E}}{(d-1)[1+(d^2-1)\gamma]^\abs{E}}
	\bigotimes_{e \in T} \Phi_e\\
    &+ \frac{1}{[1+(d^2-1)\gamma]^\abs{E}}
	\sum_{\emptyset \ne T\subseteq E} \gamma^{\abs{T}} \rho_T \otimes \bigotimes_{e
	\notin T} \Phi_e.
\end{aligned}
\end{equation}
Since $\rho_T \otimes \bigotimes_{e \notin T} \Phi_e$ is PPT with respect to
$\Gamma_T$, and the cut space $\mathcal{B}(G)$ of tree networks contains all
subsets of $E$, we can find a bipartition $S|S^c$ such that $T(S)=T$, and thus
$\mathcal{N}_S \qty(\frac{\rho_T}{\Tr(\rho_T)} \otimes \bigotimes_{e \notin T}
\Phi_e)=0$ and $\sum_{\emptyset \ne T\subseteq E} \gamma^{\abs{T}} \rho_T
\otimes \bigotimes_{e \notin T} \Phi_e$ is a PPT mixture.
Note that the coefficient of the first term on the right-hand side of
Eq.~\eqref{eq:IPEN_decompose} is exactly $-\frac{2\Tr(W_E\rho)}{d-1}$.
Therefore, when $\Tr(W_E\rho)=0$, i.e.,
\begin{equation}
\frac{2d-(d+1)[1+(d-1)\gamma]^\abs{E}}{(d-1)[1+(d^2-1)\gamma]^\abs{E}}=0,
\end{equation}
$\rho$ is a PPT mixture.
We denote by $p_0$ the critical visibility such that $\Tr(W_E\rho)=0$, and we
have
\begin{equation}
p_0=\frac{1-\gamma_0}{1+(d^2-1)\gamma_0} \qq{with} \gamma_0
=\frac{\sqrt[\abs{E}]{\frac{2d}{d+1}}-1}{d-1}.
\end{equation}
When $\Tr(W_E\rho)>0$, or equivalently, $p<p_0$, the state $\rho=\bigotimes_{e \in E} \rho_e^p$ can be obtained from $\bigotimes_{e \in E} \rho_e^{p_0}$ by local operations on each party and classical communication, hence is also a PPT mixture.

Moreover, we can prove that $-\Tr(W_E\rho)$ is exactly the GME negativity, i.e., $-\Tr(W_E\rho)=\mathcal{N}(\rho)$, when $\Tr(W_E\rho) < 0$.
We have just proved that there is a bipartition $S|S^c$ such that $\mathcal{N}_S(\rho_T \otimes \bigotimes_{e \notin T} \Phi_e)=0$, then from Eq.~\eqref{eq:GME_negativity} and Eq.~\eqref{eq:IPEN_decompose}, the GME negativity of $\rho$ is upper bounded by
\begin{equation}
	\mathcal{N}(\rho) \le \frac{2d-(d+1)[1+(d-1)\gamma]^\abs{E}}{(d-1)[1+(d^2-1)\gamma]^\abs{E}} \mathcal{N}_{S'}\qty(\bigotimes_{e \in E} \Phi_e)
\end{equation}
for arbitrary nonempty $S' \subseteq V$.
Selecting $S'$ such that $T(S')=\qty{e}$ for some $e$, calculations lead to that
\begin{equation}
	\mathcal{N}(\rho) \le d\qty[\tfrac{1+(d^2-1)p}{d^2}]^\abs{E} - \tfrac{d+1}{2} \qty[\tfrac{1+(d-1)p}{d}]^\abs{E} = -\Tr(W_E \rho).
\end{equation}
Recall that Eq.~\eqref{eq:tree_witness_value} leads to that
$\mathcal{N}(\rho)\ge-\Tr(W_E \rho)$, we thus prove that
$\mathcal{N}(\rho)=-\Tr(W_E\rho)$.
\end{proof}

\subsection{Biseparability of IPEN states in the bilocal network}
We have shown that a IPEN state $\rho=\bigotimes_{e \in E}\rho_e^p$ in tree networks have non-zero negativity, if and only if $\Tr(W_E\rho)<0$. 
For $\rho$ in the bilocal network in Fig.~\ref{fig:networks}(a), it is further the necessary and sufficient condition for GME, that is, $\rho$ is GME if and only if $\Tr(W_E\rho)<0$.

Still, we let $\gamma=\frac{1-p}{1+(d^2-1)p}$
(i.e., $\rho^p=\frac{\Phi+\gamma(\mathbb{I}-\Phi)}{1+(d^2-1)\gamma}$) and hence
\begin{equation}
\rho=\frac{[\Phi+\gamma(\mathbb{I}-\Phi)]^{\otimes 2}}{[1+(d^2-1)\gamma]^2}.
\end{equation}
Then we have $\Tr(W_E\rho) < 0$ when $\gamma < \gamma_0
=\frac{\sqrt{\frac{2d}{d+1}}-1}{d-1}$. When $\gamma = \gamma_0$, $\rho$ can be
written in the form
\begin{equation}
    \begin{aligned}
        	\rho = &\gamma  \frac{(\mathbb{I}+d\Phi) \otimes \Phi + \Phi \otimes (\mathbb{I}+d\Phi)}{[1+(d^2-1)\gamma]^2}\\
            + &\gamma^2 \frac{(\mathbb{I}-\Phi)^{\otimes 2}+(d^2-1)\Phi^{\otimes 2}}{[1+(d^2-1)\gamma]^2}.
    \end{aligned}
\end{equation}
It is known that $\mathbb{I}+\beta\Phi$ is separable when $\beta \in [-1, d]$ \cite{TerhalSchmidtnumberdensity2000}. Therefore, $\mathbb{I}+d\Phi$ is separable, and thus $(\mathbb{I}+d\Phi)\otimes\Phi$ and $\Phi\otimes(\mathbb{I}+d\Phi)$ is biseparable. 
For the state $(\mathbb{I}-\Phi)^{\otimes 2}+(d^2-1)\Phi^{\otimes 2}$, note that 
\begin{equation}
    (\mathcal{T}_{(A, B_1)}\otimes\mathcal{T}_{(B_2,C)})(\Phi_{AC}\otimes\Phi_{B_1B_2})= \tfrac{(\mathbb{I}-\Phi)^{\otimes 2}+(d^2-1)\Phi^{\otimes 2}}{d^2(d^2-1)},
\end{equation}
where $\mathcal{T}_{(A, B_1)}$ and $\mathcal{T}_{(B_2,C)}$ are both $U\otimes U^*$-twirling operations and thus can be performed through local operations and classical communications. 
Since $\Phi_{AC}\otimes\Phi_{B_1B_2}$ is separable with respect to the partition $B|AC$, the state $(\mathbb{I}-\Phi)^{\otimes 2}+(d^2-1)\Phi^{\otimes 2}$ is also separable with respect to $B|AC$.
Therefore, $\rho$ is also biseparable when $\gamma=\gamma_0$.
For $\gamma>\gamma_0$, $\rho$ can be obtained by local operations and classical communications from $\rho$ with $\gamma=\gamma_0$, then is also biseparable.

\subsection{Biseparability of IPEN states in the triangle network}
For the IPEN state $\rho=\bigotimes_{e \in E} \rho_e^p$  in the triangle network in Fig.~\ref{fig:networks}(d), the SSW also provides necessary and sufficient condition for GME.
The SSW is
\begin{equation}\label{eq:triangle}
    W^\Delta = -\frac{d^2-1}{2}\bigotimes_{e \in E} \Phi_e + \frac{1}{2} \sum_{\emptyset \ne T \subsetneq E} \bigotimes_{e \in T} \Phi_e \otimes \bigotimes_{e \notin T} (\mathbb{I}_e-\Phi_e).
\end{equation}
Moreover, this SSW also quantifies the negativity of $\rho$, that is, $\mathcal{N}(\rho)=\max\qty{-\Tr(W^\Delta \rho),0}$.

The proof is as following. First, we prove that $W^\Delta$ is an SSW satisfying Eq.~\eqref{eq:GME_negativity_dual}, hence $\rho$ is GME when $\Tr(W^\Delta \rho) < 0$ and $-\Tr(W^\Delta \rho)$ provides an lower bound of $\mathcal{N}(\rho)$, i.e., $-\Tr(W^\Delta \rho) \le \mathcal{N}(\rho)$.
Next, we show that when $\Tr(W^\Delta \rho) \ge 0$, $\rho$ is biseparable, thus $W^\Delta$ provides the necessary and sufficient condition for GME in $\rho$.
Finally, it is proved that when $\Tr(W^\Delta \rho) < 0$, we can find a decomposition $\sum_S p_S \rho_S$ of $\rho$ satisfying constraints in Eq.~\eqref{eq:GME_negativity} such that $\sum_S p_S \mathcal{N}_S(\rho_S)=-\Tr(W^\Delta \rho)$, which is an upper bound of $\mathcal{N}(\rho)$, that is, $-\Tr(W^\Delta \rho) \ge \mathcal{N}(\rho)$. It follows that $\mathcal{N}(\rho)=-\Tr(W^\Delta \rho)$.

\begin{proof}
For simplicity, we denote the three parties by $A$, $B$, and $C$ with $E=\qty{(A,B),\ (B,C),\ (C, A)}$, and use $X \otimes Y \otimes Z$ to denote $X_{(A, B)} \otimes Y_{(B, C)} \otimes Z_{(C, A)}$.
To verify that $W$ is an SSW, note that $W \succeq Q_B^{\Gamma_B}$, where $Q_B=P^+\otimes P^-\otimes\Phi + P^-\otimes P^+\otimes\Phi$ and $0 \preceq Q_B \preceq \mathbb{I}$.
The symmetry of $W$ leads to similar inequalities $W \succeq Q_A^{\Gamma_A}$ and $W \succeq Q_C^{\Gamma_C}$.

Still, we let $\gamma=\frac{1-p}{1+(d^2-1)p}$
(i.e., $\rho^p=\frac{\Phi+\gamma(\mathbb{I}-\Phi)}{1+(d^2-1)\gamma}$).
Then we have
\begin{equation}
	\Tr(W^\Delta \rho) = \frac{d^2-1}{2[1+(d^2-1)\gamma]^3} \qty[3(d^2-1)\gamma^2 + 3\gamma -1].
\end{equation}
$W^\Delta$ can certify GME in $\rho$, i.e., $\Tr(W^\Delta \rho)<0$ when $3(d^2-1)\gamma^2+3\gamma-1<0$, which holds when 
\begin{equation}
	\gamma < \gamma_0 := \frac{\sqrt{3(4d^2-1)}-3}{6(d^2-1)}.
\end{equation}

Then we want to show the biseparability of $\rho$ when $\gamma \ge \gamma_0$.
When $\gamma \ge \frac{1}{d+1}$, $\rho$ is fully separable since it is the tensor product of separable states $\Phi+\gamma(\mathbb{I}-\Phi)$. 
Thus we only need to check the case that $\gamma_0 \le \gamma < \frac{1}{d+1}$.
To prove biseparability of $\rho$, we decompose it as $\rho=\rho_A+\rho_B+\rho_C$, where
\begin{equation}\label{eq:rhob}
    \rho_B = \frac{1}{[1+(d^2-1)\gamma]^3} \qty[\sigma\otimes\Phi + \frac{\gamma^3}{3} (\mathbb{I}-\Phi)^{\otimes3}]
\end{equation}
and
\begin{equation}
    \sigma = \tfrac{\Phi^{\otimes 2}}{3} + \tfrac{\gamma}{2} [(\mathbb{I}-\Phi)\otimes\Phi + \Phi\otimes(\mathbb{I}-\Phi)] + \gamma^2(\mathbb{I}-\Phi)^{\otimes 2}.
\end{equation}
$\rho_A$ and $\rho_C$ can be obtained from symmetry of the state.
In the following, we prove that $\rho_B$ is separable with respect to the partition $B|AC$ when $\gamma \ge \gamma_0$, and thus $\rho$ is biseparable.

The state $\mathbb{I}-\Phi$ is a separable bipartite state, therefore $(\mathbb{I}-\Phi)^{\otimes 3}$ is separable with respect to the partition $B|AC$. 
Then we only need to check when $\sigma \otimes \Phi$ is separable with respect to $B|AC$, or more specifically, when $\sigma$ is separable with respect to $B|AC$.
Since a state is separable if and only if its partial transpose is separable, we investigate the separability of $\sigma^{\Gamma_B}$ instead.
We have
\begin{equation}
	\sigma^{\Gamma_B} = \lambda_1 {P^+}^{\otimes 2} + \lambda_2 {P^-}^{\otimes 2} + \lambda_3 (P^+ \otimes P^- + P^- \otimes P^+),
\end{equation}
where 
\begin{equation}
	\begin{aligned}
		&\lambda_1 = \frac{3(d-1)^2\gamma^2 + 3(d-1)\gamma + 1}{3d^2} > 0,\\
		&\lambda_2 = \frac{3(d+1)^2\gamma^2 - 3(d+1)\gamma + 1}{3d^2} > 0,\\
		&\lambda_3 = \frac{3(d^2-1)\gamma^2 + 3\gamma - 1}{3d^2}.
	\end{aligned}
\end{equation}
When $\gamma \ge \gamma_0$, it also holds $\lambda_3\ge0$.
Since $P^+ + P^- = \mathbb{I}$, we can write $\sigma^{\Gamma_B}$ in the form of
\begin{equation}
	\sigma^{\Gamma_B}=
	(\lambda_1-\lambda_2){P^+}^{\otimes 2} 
    + (\lambda_2-\lambda_3)({P^+}^{\otimes 2}
    + {P^-}^{\otimes 2}) + \lambda_3\mathbb{I}^{\otimes 2},
\end{equation}
when $\lambda_2\ge \lambda_3$, and
\begin{equation}
	\begin{aligned}
		\sigma^{\Gamma_B} = &
		(\lambda_1+\lambda_2-2\lambda_3){P^+}^{\otimes 2}\\
		&+ (\lambda_3-\lambda_2) (P^+\otimes \mathbb{I} + \mathbb{I} \otimes P^+) + \lambda_2\mathbb{I}^{\otimes 2}.
	\end{aligned}
\end{equation}
when $\lambda_2 < \lambda_3$, where $\lambda_1-\lambda_2=2\gamma(1-2\gamma)/d\ge0$ when $\gamma < \frac{1}{d+1}$ and $ \lambda_1 + \lambda_2 - 2 \lambda_3 = \frac{4}{d^2}(\gamma^2-\gamma+1/3) > 0$.
Moreover, we can prove that the states ${P^+}^{\otimes 2}$, ${P^+}^{\otimes 2} + {P^-}^{\otimes 2}$, and $P^+\otimes\mathbb{I} + \mathbb{I}\otimes P^+$ are all separable with respect to $B|AC$, and so is $\sigma^{\Gamma_B}$ and $\sigma$. 
First, since ${P^+}^{\Gamma}=\frac{\mathbb{I}+d\Phi}{2}$ are separable, $P^+$ is also separable. Hence ${P^+}^{\otimes 2}$ and $P^+\otimes \mathbb{I} + \mathbb{I} \otimes P^+$ are separable with respect to $B|AC$.
On the other hand, note that $\qty({P^+}^{\otimes 2} + {P^-}^{\otimes 2})^{\Gamma_B} = \frac{\mathbb{I}^{\otimes 2} + d^2 \Phi^{\otimes 2}}{2}$, which is separable with respect to $B|AC$, as $\Phi^{\otimes 2}$ is the maximally entangled state between $B$ and $AC$, whose local dimension is $d^2$. Therefore, ${P^+}^{\otimes 2} + {P^-}^{\otimes 2}$ is also separable with respect to $B|AC$.

To calculate the GME negativity $\mathcal{N}(\rho)$, recall that $\rho=\rho_A + \rho_B + \rho_C$ defined in Eq.~\eqref{eq:rhob}.
When $\gamma < \gamma_0$, we have
\begin{equation}
	\mathcal{N}(\rho) \le \sum_{S = A, B, C}\Tr\rho_S \mathcal{N}_S \qty(\frac{\rho_S}{\Tr\rho_S}) = -\Tr(W^\Delta \rho).
\end{equation}
Since $W^\Delta$ is an SSW satisfying Eq.~\eqref{eq:GME_negativity_dual}, we also have $\mathcal{N}(\rho) \ge -\Tr(W^\Delta \rho)$.
Therefore, we have $\mathcal{N}(\rho)=-\Tr(W^\Delta \rho)$.
\end{proof}

\subsection{Highly connected networks}\label{app:BG-SSW}
In Observation~\ref{obs:asymptotic}, we consider the sequence of networks such that asymptotic degree of every party exceeds half of the network size, i.e., $\lim_{\abs{V}\to\infty}\frac{\deg(v)}{\abs{V}}>\frac{1}{2}$ for
every $v\in V$.
Here, we employ
\begin{equation}
	W_{\mathcal{B}(G)} = -\bigotimes_{e\in E}\Phi_e
    + \sum_{\emptyset \ne K \in \mathcal{B}(G)}
    \bigotimes_{e\in K}\frac{\mathbb{I}_e-\Phi_e}{d-1}
    \otimes \bigotimes_{e\notin K}\Phi_e
\end{equation}
in Eq.~\eqref{eq:BG-SSW} to verify GME.

For the IPEN state $\rho=\bigotimes_{e\in E}\rho_e^p=\bigotimes_{e \in E} \frac{\Phi_e+\gamma(\mathbb{I}_e-\Phi_e)}{1+(d^2-1)\gamma}$ in the networks, the GME can be certified if 
\begin{equation}
	\Tr(W_{\mathcal{B}(G)} \rho) = \frac{\qty(\sum_{\emptyset \ne K \in \mathcal{B}(G)} [\gamma(d+1)]^{\abs{K}})-1}{[1+(d^2-1)\gamma]^\abs{E}} < 0.
\end{equation}
When $\gamma \ge \frac{1}{d+1}$, i.e., $\gamma(d+1) \ge 1$, the state $\rho^p_e=\frac{\Phi+\gamma(\mathbb{I}-\Phi)}{1+(d^2-1)\gamma}$ is separable, hence $\rho$ is biseparable. Therefore, we only consider the case that $\rho^p_e$ is entangled, i.e., $\gamma(d+1)<1$.
In this scenario, we can show that for highly connected networks such that $\lim_{\abs{V}\to\infty}\frac{\deg(v)}{\abs{V}}>\frac{1}{2}$ for any $v \in V$, it holds that
\begin{equation}
    \lim_{\abs{V}\to\infty} \sum_{\emptyset\ne K \in \mathcal{B}(G)} [\gamma(d+1)]^\abs{K} = 0,
\end{equation}
thus $\Tr(W_{\mathcal{B}(G)}\rho) < 0$ always holds for sufficiently large $\abs{V}$ as long as $\rho_e^p$ is entangled.
Then GME can be certified and Observation~\ref{obs:asymptotic} follows.
This is a generalization of the asymptotic survival of GME in complete networks in \cite{ContrerasTejadaAsymptoticSurvivalGenuine2022}, which shows that there is a general bound of $\gamma$ such that the GME asymptotically survives.

\begin{proof}

For the networks under consideration, there exists real number $c>\frac{1}{2}$
such that all vertices $v$ have degree $\deg(v)\ge c\abs{V}>\frac{\abs{V}}{2}$,
provided the network size $\abs{V}$ is large enough.
Let $\abs{V}=2k+1$ be odd (for even $\abs{V}$, similar discussion will give the
same result). For any $i=1,\dots,k$, there are $\binom{\abs{V}}{i}$ different
bipartition of $V$, each of which results in a cut $K$ with $\abs{K}> i(2c-1)k$.
To see this, we consider a bipartition $S|S^c$ where $S$ has $i$ parties ($i\le k$) without loss of generality.
On the one hand, as there are only $i$ parties in $S$, for any $v\in S$, there exist at most $i-1$ vertices in $S$ adjacent to $v$.
On the other hand, for each $v\in S$, since $\deg(v)\ge c\abs{V}$, there should be at least $c\abs{V}$ vertices in $V$ adjacent to $v$.
Therefore, at least $c\abs{V}-(i-1)$ edges incident with $v$ are in the cut of $S|S^c$, where
\begin{equation}
    c\abs{V}-(i-1)>(c-\frac{1}{2})\abs{V}>(2c-1)k.
\end{equation}
Hence we have
\begin{equation}
	\begin{aligned}
		&\sum_{\emptyset \ne K \in \mathcal{B}(G)} [\gamma(d+1)]^\abs{K}\\
		< & \sum_{i=1}^k \binom{\abs{V}}{i} [\gamma(d+1)]^{i(2c-1)k}\\
		\le & \sum_{i=1}^k (2k+1)^i [\gamma(d+1)]^{i(2c-1)k}.
	\end{aligned}
\end{equation}
Let $f_k:=(2k+1)[\gamma(d+1)]^{(2c-1)k}$, and it holds that $\lim_{k \to \infty} f_k = 0$ when $\gamma(d+1)<1$.
We have (when $f_k\ne 1$)
\begin{equation}
	\sum_{\emptyset \ne K \in \mathcal{B}(G)} [\gamma(d+1)]^{\abs{K}} \le f_k \frac{1-f_k^k}{1-f_k},
\end{equation}
which converges to $0$ as $k\to \infty$ ($\abs{V}\to\infty$).
Therefore, we have $\lim_{\abs{V}\to\infty}\qty(\sum_{\emptyset \ne K \in \mathcal{B}(G)} [\gamma(d+1)]^{\abs{K}})-1 = -1 < 0$ as long as $\gamma(d+1)<1$.
This means the critical visibility $p=\frac{1-\gamma}{1+(d^2-1)\gamma}\to \frac{1}{d+1}$ as $\abs{V}\to\infty$.
\end{proof}

\section{Complexity of estimation}\label{app:complexity}
The fidelity with respect to the maximally entangled state can be measured
efficiently, based on the method of direct fidelity estimation
\cite{FlammiaDirectFidelityEstimation2011}
or quantum state verification \cite{ZhuOptimalverificationfidelity2019,
YuOptimalverificationgeneral2019, YuStatisticalMethodsQuantum2022}.
Notably, required copies of the states in these methods are known not to increase with the system dimension.
Here, we provide a simple analysis based on the stabilizer formalism.

It is easy to verify that the maximally entangled state $\ket{\phi^+}$ is the only state stabilized by $X \otimes X$ with $X\ket{j}=\ket{j+1}$ and $Z \otimes Z^\dagger$ with $Z\ket{j} = \mathrm{e}^{\frac{2\pi\mathrm{i}}{d}j}\ket{j}$, where the addition is modulo $d$, that is, $X\otimes X\ket{\phi^+}=Z\otimes Z^\dagger \ket{\phi^+}=\ket{\phi^+}$.
Note that $X \otimes X$ and $Z \otimes Z^\dagger$ commutes with each other.
Denoting $X^i Z^j\otimes X^i {Z^\dagger}^j$ by $M_{i, j}$, we have $\Phi = \ketbra{\phi^+} = \frac{1}{d^2} \sum_{i,j=0}^{d-1}  M_{i, j}$, thus the fidelity with respect to $\Phi$ can be estimated by
\begin{equation}
	\Tr(\Phi \rho) = \frac{1}{d^2} \sum_{i,j=0}^{d-1} \Tr(M_{i, j} \rho).
\end{equation}
$M_{i, j}$ is not Hermitian in general, thus cannot be measured directly. 
However, by noting that $M_{i, j} + M_{d-i, d-j}$ is Hermitian, we can measure it instead, and it can also be measured locally.
Let $X^i Z^j = (Z^{d-j}X^{d-i})^{-1}=\sum_{k=0}^{d-1} \mathrm{e}^{i \varphi_k} \ketbra{\psi_k}$ and $X^i {Z^\dagger}^j = ({Z^\dagger}^{d-j}X^{d-i})^{-1}=\sum_{k=0}^{d-1} \mathrm{e}^{i \varphi'_k} \ketbra{\psi'_k}$.
Then we have
\begin{equation}
    \begin{aligned}
        &M_{i, j} +M_{d-i, d-j}\\
        =&X^iZ^j\otimes X^i {Z^\dagger}^j + Z^{d-j}X^{d-i}\otimes {Z^\dagger}^{d-j}X^{d-i}\\
        =& 2 \sum_{k,l=0}^{d-1} \cos(\varphi_k+\varphi'_l) \ketbra{\psi_k} \otimes \ketbra{\psi'_l}.
    \end{aligned}
\end{equation}
Note that the eigenvalues and eigenvectors can be solved analytically, making the measurements feasible.
In this case, we can measure the observable $\frac{\mathbb{I}}{d^2}+\frac{d^2-1}{2d^2}(M_{i,j}+M_{d-i,d-j})$ with probability of $\frac{1}{d^2-1}$ (for $i\ne 0$ or $j \ne 0$), then the expected value is $\Tr(\Phi \rho)$.

Then we can use this to estimate SSWs.
Let $X_1,\dots,X_m$ be independent random variables such that $a_i\le X \le b_i$, and $S_m = \frac{X_1 + \dots + X_m}{m}$, the Hoeffding's inequality \cite{HoeffdingProbabilityInequalitiesSums1963} states that
\begin{equation}\label{eq:Hoeffding}
	\Pr(S_m-\mathbb{E}(S_m) \ge \epsilon) \le \exp(-\frac{2m^2\epsilon^2}{\sum_{i=1}^m(b_i-a_i)^2}).
\end{equation}
Therefore, the complexity of the estimation is determined by the maximal and minimal value of the random variables.
For $W_E$ in Eq.~\eqref{eq:E-SSW}, we renormalize it as
\begin{equation}
	\tilde{W}_E = -\frac{2d}{d-1} \bigotimes_{e \in E} \Phi_e + \frac{d+1}{d-1} \bigotimes_{e \in E} \qty(\frac{\mathbb{I}_e+d\Phi_e}{d+1})
\end{equation}
such that $\max_\rho \abs{\Tr(\tilde{W}_E \rho)}=1$ to provide a reasonable estimation.
Then the measurement result is
\begin{equation}
	-\frac{2d}{d-1} \prod_{i=1}^\abs{E} m_i + \frac{d+1}{d-1} \prod_{i=1}^\abs{E} \frac{1+dm_i}{d+1},
\end{equation}
where $m_i$ is the measurement result on the $i$th edge.
When we measure $\frac{\mathbb{I}}{d^2}+\frac{d^2-1}{2d^2}(M_{i,j}+M_{d-i,d-j})$ on the $i$th edge, it holds that $-1 + \frac{2}{d^2} \le m_i \le 1$. Then the measurement result for $\tilde{W}_E$ is lower bounded by $-\frac{2d}{d-1}$, and upper bounded by $\frac{2(d^2-2)}{d(d-1)}$.
Substituting them into Eq.~\eqref{eq:Hoeffding}, we can get
\begin{equation}
	\Pr(\bar{W}_E - \expval{\tilde{W}_E} \le -\epsilon) \le \exp(-\frac{m\epsilon^2}{8\qty(1+\frac{1}{d})^2}),
\end{equation}
where $\bar{W}_E$ is the average value of the $m$ measurements.
By letting $-\epsilon=\bar{W}_E$ when $\bar{W}_E<0$, we get 
\begin{equation}
	\Pr(\Tr(\tilde{W}_E\rho) \ge 0) 
	\le \exp(-\tfrac{m\bar{W}_E^2}{8(1+\tfrac{1}{d})^2}) 
	\le \exp(-\tfrac{m\bar{W}_E^2}{18}),
\end{equation}
which is independent of the local dimensions and the scales of the networks.

For general SSWs, while required copies of the states for fidelity estimation do not increase with the local dimension, the sample complexity depends on the network's topology, which makes the general analytical discussion difficult.
However, numerical simulation suggests that sometimes it may scale exponentially with the scale of the networks.
In practice, there are two strategies to mitigate this exponential complexity.
First, leveraging joint measurements—specifically only using
two-party measurements with respect to the maximally entangled state—ensures that the sample complexity is independent of the number of parties.
Second, adopting the stronger structural assumption that the network state
adheres to a tensor product structure (as shown in Fig.~\ref{fig:networks})
simplifies the measurement task to only estimating the local fidelities
$\Tr(\Phi_e\rho_e)$.

\end{document}